\documentclass{article}
\usepackage{graphicx} 
\usepackage[english]{babel}
\usepackage[a4paper,top=2cm,bottom=2cm,left=3cm,right=3cm,marginparwidth=1.75cm]{geometry}
\usepackage[colorlinks=true, allcolors=blue]{hyperref}

\usepackage{soul}
\usepackage{url}
\usepackage[utf8]{inputenc}
\usepackage[small]{caption}
\usepackage{subcaption}
\usepackage{amsmath,amssymb,amsthm}
\usepackage{amsfonts}
\usepackage{booktabs}
\usepackage{color}
\usepackage{verbatim}
\usepackage[ruled,linesnumbered]{algorithm2e}
\usepackage{xfrac}
\usepackage{float}
\usepackage{cleveref}
\usepackage{dsfont}
\usepackage{multicol}
\usepackage{pgfplots}						
\usepackage{tikz}
\pgfplotsset{compat=1.17}
\usetikzlibrary{patterns}
\usepackage{pgfplotstable}
    \pgfplotsset{
    name nodes near coords/.style={
        every node near coord/.append style={
            name=#1-\coordindex,
            alias=#1-last,
        },
    },
    name nodes near coords/.default=coordnode
    }

\newtheorem{claim}{Claim}[section]
\newtheorem{theorem}{Theorem}[section]

\newtheorem{lemma}[theorem]{Lemma}
\newtheorem{observation}[theorem]{Observation}
\newtheorem{example}[theorem]{Example}
\newtheorem{definition}[theorem]{Definition}

\newenvironment{proofof}[1]{{\vspace*{5pt} \noindent\bf Proof of #1:  }}{\hfill\rule{2mm}{2mm}\vspace*{5pt}}

\makeatletter
\newtheorem*{rep@theorem}{\rep@title}
\newcommand{\newreptheorem}[2]{%
\newenvironment{rep#1}[1]{%
 \def\rep@title{#2 \ref{##1}}%
 \begin{rep@theorem}}%
 {\end{rep@theorem}}}
\makeatother

\newreptheorem{lemma}{Lemma}
\newreptheorem{claim}{Claim}

\newcommand{\bR}{{\mathbb{R}}}

\newcommand{\MMS}{\mathsf{MMS}}
\newcommand{\bX}{\mathbf{X}}

\newcommand{\bc}{\mathbf{c}}

\newcommand{\cI}{\mathcal{I}}

\DeclareMathOperator*{\argmax}{argmax}
\DeclareMathOperator*{\argmin}{argmin}

\title{Guaranteeing MMS for All but One Agent When Allocating Indivisible Chores}

\author{
Jiawei Qiu\thanks{IOTSC, University of Macau. \{mc25118,xiaoweiwu,yc27429,yc17423\}@um.edu.mo. The authors are ordered alphabetically.}
\and
Xiaowei Wu $^*$
\and 
Cong Zhang $^*$
\and 
Shengwei Zhou $^*$
}

\begin{document}

\maketitle

\begin{abstract}
We study the problem of allocating $m$ indivisible chores to $n$ agents with additive cost functions under the fairness notion of maximin share (MMS). 
In this work, we propose a notion called $\alpha$-approximate all-but-one maximin share ($\alpha$-AMMS) which is a stronger version of $\alpha$-approximate MMS. 
An allocation is called $\alpha$-AMMS if $n-1$ agents are guaranteed their MMS values and the remaining agent is guaranteed $\alpha$-approximation of her MMS value. 
We show that there exist $\alpha$-AMMS allocations, with $\alpha = 9/8$ for three agents; $\alpha = 4/3$ for four agents; and $\alpha = (n+1)^2/4n$ for $n\geq 5$ agents.

\end{abstract}

\section{Introduction}
Since 1948 when Steinhaus~\cite{steinhaus1948problem} proposed the fair allocation problem, it has been a central problem in the area of computer science, economics, and social choice theory.
Mainstream research on the allocation area can be characterized into two categories, one is resource allocation where items are goods, and the other is task allocation where items are chores.

In this work, we consider the chores allocation of assigning a set of $m$ indivisible chores $M$ to a set of $n$ agents.
For each agent $i\in N$, we use $c_i: 2^M \xrightarrow{} \bR^+ \cup \{ 0 \}$ to denote her cost function, which is additive, i.e., for any bundle $S \subseteq M, c_i(S) = \sum_{e \in S}c_i(e)$. 
We consider the discrete setting where each chore needs to be allocated to exactly one agent.
When items are divisible, a natural and well-studied criterion is \emph{proportionality}~\cite{steinhaus1948problem} that requires each agent to receive at most her proportional cost of all chores.
While proportionality is no longer guaranteed to exist in the discrete setting\footnote{Considering allocating one item to two agents.}, some refined criteria have received much attention.
Among all the notions, maximin share (MMS) is one of the most popular criteria since introduced by Budish $et$ $al.$~\cite{conf/bqgt/Budish10}.
For each agent $i$, her maximin share $\MMS_i$ is defined as the minimum cost she can guarantee if she divides all the items into $n$ bundles and picks a bundle at last.
When there are only two agents, the definition of MMS coincides with the well-studied mechanism \emph{divide-and-choose}.
Unfortunately, when there are more than two agents, MMS allocations are not guaranteed to exist for both goods and chores~\cite{journals/jacm/KurokawaPW18,conf/wine/FeigeST21,conf/aaai/AzizRSW17}.
Much research tends to consider relaxations such as the approximation of MMS.

\paragraph{Approximate MMS.}
While no algorithm can achieve approximation ratios smaller than 44/43 for chores~\cite{conf/wine/FeigeST21}, the very first upper bound of the approximation ratio is $2- 1/n$ shown by~\cite{conf/aaai/AzizRSW17} where the algorithm only involves ordinal preference on items.
The result is further improved to $5/3$ by Aziz $et$ $al.$~\cite{journals/mp/AzizLW24} and then to $3/2$ by Feige and Huang~\cite{conf/sigecom/FeigeH23}.
When considering the cardinal information, Barman and Krishnamurthy~\cite{journals/teco/BarmanK20} showed the existence of $4/3$ approximation of MMS, which was further improved to $11/9$ in~\cite{conf/sigecom/HuangL21}, and then to $13/11$ by Huang and Segal-Halevi~\cite{conf/sigecom/HuangS23}.
For the analogous problem of goods, there are also much research working on the approximation of MMS~\cite{journals/jacm/KurokawaPW18,journals/teco/BarmanK20,journals/mor/GhodsiHSSY21,journals/ai/GargT21,conf/ijcai/AkramiGST23}, which lead to the state-of-the-art ratio of $3/4 + 3/3836$~\cite{conf/soda/AkramiG24}.

\medskip

While the above works focused on guaranteeing all agents a relaxed threshold (rather than MMS), for the allocation of goods, there is a line of research focusing on guaranteeing a fraction of agents their MMS.
The goal is to investigate the maximum fraction of agents that can be guaranteed.
To evaluate the allocation, Hosseini and Searns~\cite{conf/ijcai/HosseiniS21} proposed the framework of $(\alpha,\beta)$-MMS, that guarantees $\alpha$-fraction of agents $\beta$-approximation of MMS. For the remaining $(1-\alpha)$-fraction of agents, there is no guarantee. They show that $(\frac{n-1}{n},1)$-MMS allocations always exist for $n\le 4$ and this is the first time that AMMS allocation was studied in fair allocation of indivisible goods setting
When applying the $(\alpha, \beta)$-MMS framework to the chores setting, guaranteeing most agents their MMS value is trivial and might lead to extreme, e.g., allocating all chores to one single agent is a $(\frac{n-1}{n}, 1)$-MMS allocation. 
We study the approximation of MMS allocation from another angle and introduce a framework for the chores allocation called \emph{all-but-one maximin share} (AMMS).
An allocation is called $\alpha$-AMMS if there are $n-1$ agents guaranteed their MMS and the remaining agent is guaranteed $\alpha$-approximation of her MMS.
A natural question that we want to ask is
\begin{center}
    \textit{When ensuring MMS for $n-1$ agents, what is the smallest approximation of MMS we can guarantee for the remaining agent?}
\end{center}

In contrast to the $(\frac{n-1}{n},1)$-MMS framework, which does not provide any guarantee for the single agent who is not guarantee her MMS, the AMMS framework aims to establish an upper bound on the cost incurred by that agent.
Besides, the traditional approximate-MMS framework can not imply any results of AMMS allocations due to the indivisibility of items. 
By adopting an alternative perspective on the (approximation of) MMS allocations, the AMMS framework offers a potentially different understanding of the MMS concept.

\subsection{Our Results}
In this paper, we consider the existence of $\alpha$-AMMS allocations for indivisible chores while minimizing the approximation guarantee for the last agent.
We first consider instances with three agents, for which there exists a hardness showing that no algorithm can guarantee $\alpha$-approximation of MMS allocations while $\alpha < 44/43$~\cite{conf/wine/FeigeST21}. This hardness also provides a lower bound of $44/43$ with respect to AMMS allocations for three agents.
For the upper bound, we show the existence of $9/8$-AMMS allocations for three agents.
We further extend our result to instances with four agents and show the existence of $4/3$-AMMS allocations.

\smallskip
\noindent
{\bf Result 1} (Theorem~\ref{thm:three-agents}) {\bf .}
{\em For any instance with three agents, there exist $9/8$-AMMS allocations.}
\smallskip

\noindent
{\bf Result 2} (Theorem~\ref{theorem:fourAgents}) {\bf .}
{\em For any instance with four agents, there exist $4/3$-AMMS allocations.}
\smallskip

We show the existence of $9/8$-AMMS allocations for three agents, by revisiting and investigating some classic procedures, e.g., \emph{divide-and-choose} and \emph{load-balancing}.
We further show that our analysis is tight, by providing an instance in which our algorithm returns a $9/8$-AMMS allocation.
We then move to four agent instances and show the existence of $4/3$-AMMS allocations.
Part of our analysis of Result 2 is inspired by the~\emph{atomic bundle} framework~\cite{conf/wine/FeigeST21,journals/corr/abs-2205-05363} that involves the intersection of two bundles, each of which is from a different agent's MMS partitions.
By applying the~\emph{atomic bundle} framework, we characterize cases with nice properties, for which we extend the classic \emph{load-balancing} procedure to the four agent instances.

Finally, we consider instances with a general number of agents.
We present an upper bound of $(n+1)^2/4n$ for the approximation ratio of AMMS allocations. 
Our technique is based on a matching framework between agents and bundles and involves the reduction from general instances to new instances with smaller sizes. 
More specifically, we construct a bipartite graph $G$ called the MMS-feasibility graph in which one side is the set of agent $N$ and the other side is an $n$-partition of items $M$.
If there exists a perfect matching in graph $G$, then there exists an MMS allocation.
Otherwise, we show that by guaranteeing some agents their MMS, we can focus on the remaining agents as a reduced instance.

\smallskip
\noindent
{\bf Result 3} (Theorem~\ref{thm:general-number-agents}) {\bf .}
{\em For any instance with $n\geq 5$ agents, there exist $\frac{(n+1)^2}{4n}$-AMMS allocations.}
\smallskip

We remark that similar connections between matching and fair allocation have also been revealed in previous research.
For example, Gan $et$ $al.$~\cite{journals/mss/GanSV19} used a matching framework to compute an envy-free assignment in house allocation problems. 
Kobayashi $et$ $al.$~\cite{conf/sagt/KobayashiMS23} used a  matching-based approach to study the EFX allocation problem of indivisible chores. Bu $et$ $al.$~\cite{journals/corr/abs-2404-18133} used a matching framework to study the problem of fair allocation with comparison-based queries.

\subsection{Related Works}

Since the notion of MMS has been well-studied, here we only introduce some closely related works. 
For the allocation of chores, 19/18-MMS allocations are known to exist for three agents, guaranteeing that two of the agents achieve their approximate MMS values ~\cite{journals/corr/abs-2205-05363}.
For four to seven agents, $20/17$-MMS allocation exists \cite{conf/sigecom/HuangS23}.
For the allocation of goods, Feige and Norkin~\cite{journals/corr/abs-2205-05363} showed the existence of $11/12$-MMS allocations for three agents, and showed an upper bound of $39/40$.
For four agents, $4/5$-MMS allocations are guaranteed to exist~\cite{journals/mor/GhodsiHSSY21, conf/sigecom/BabaioffF22}.
Apart from the cardinal approximation, several works focused on the ordinal approximation of MMS, called $1$-$out$-$of$-$d$ MMS~\cite{journals/isci/Aigner-HorevS22,conf/ijcai/HosseiniS21,journals/jair/HosseiniSS22}.
The current state-of-the-art results are $1$-$out$-$of$-$\lfloor \frac{3n}{4} \rfloor$ MMS for chores~\cite{conf/atal/HosseiniSS22} and $1$-$out$-$of$-$\lceil \frac{4n}{3} \rceil$ for goods~\cite{journals/corr/abs-2307-12916}.

MMS allocations have also been studied in some generalized settings, e.g., beyond additive functions~\cite{conf/nips/0037WZ23,journals/ai/GhodsiHSSY22,journals/ai/SeddighinS24,conf/nips/AkramiMSS23,journals/corr/abs-2303-12444}. 
Zhou $et$ $al.$~\cite{conf/icml/0002B023} studied (approximate) MMS allocations in the online setting with normalization assumption. 
Besides, there is a line of literature focusing on allocating items to strategic agents while guaranteeing (approximate) MMS allocations~\cite{conf/atal/BarmanG0KN19, conf/ijcai/AmanatidisBM16, conf/sigecom/AmanatidisBCM17, journals/mp/AzizLW24}.
For asymmetric agents, MMS inspired several fairness notions such as weighted MMS (WMMS)~\cite{journals/jair/FarhadiGHLPSSY19,conf/ijcai/0001C019} and AnyPrice Share (APS)~\cite{conf/sigecom/BabaioffEF21, conf/sigecom/FeigeH23}. 
Furthermore, Pairwise MMS~\cite{journals/teco/CaragiannisKMPS19} and Groupwise MMS~\cite{conf/aaai/BarmanBMN18} are defined to consider the fairness between pairwise or groupwise agents.

Very recently, Akrami $et$ $al.$~\cite{journals/corr/abs-2307-12916} proposed a framework called $T$-MMS that guarantees different approximation ratios for different agents. 
For a non-increasing sequence $T=(\tau_1,\dots,\tau_n)$ of numbers, if the agent at rank $i$ in the order gets a bundle of value at least $\tau_i$ times her MMS value, then such allocation is said to be a $T$-MMS allocation. 
Akrami $et$ $al.$ showed that there exist $T$-MMS allocations where $\tau_i\ge \max(\frac{3}{4}+\frac{1}{12n},\frac{2n}{2n+i-1})$ for all $i\in N$.

\paragraph{$\alpha$-AMMS Allocations for Goods} Hosseini $et$ $al.$~\cite{conf/ijcai/HosseiniS21} show that $(\frac{n-1}{n},1)$-MMS allocations always exist for $n\le 4$. When $n\ge 5$, the existence of $(\frac{n-1}{n},1)$-MMS allocations remains unknown. Furthermore, the existence of $(\frac{n-1}{n},1)$-MMS allocations can imply the existence of 1-out-of-$(n + 1)$ MMS allocations which is a major open problem for the ordinal approximation of MMS~\cite{conf/ijcai/HosseiniS21, journals/corr/abs-2307-12916}.

\section{Preliminary}
We consider allocating a set of $m$ indivisible chores to $n$ agents. 
We use $M = \{e_1, \ldots, e_m\}$ to denote the set of items and $N = \{1, \ldots , n\}$ to denote the set of agents, respectively.
Every agent $i\in N$ has an additive cost function $c_i : 2^M \to \bR^+ \cup \{ 0 \}$. 
An instance is denoted by $\cI = (M, N, \bc)$, where $\bc = (c_1, \ldots, c_n)$ is the set of cost functions. 
A $k$-partition of items $M$ is denoted by $P = \{ P_1, \ldots, P_k \}$ where $P_i \cap P_j = \emptyset$ for all $i \neq j$ and $\cup_{1 \leq i \leq k}P_i = M$. 
Specifically, we call an $n$-partition an allocation $\bX = (X_1, \ldots, X_n)$ if $X_i$ is the bundle allocated to agent $i$. 
For convenience of notation, we use $X_i+e$ to denote $X_i\cup \{e\}$ and $X_i-e$ to denote $X_i\setminus \{e\}$. 
We use $[k]$ to denote $\{1,\dots,k\}$.
Given any $k$, we let $\Pi_k(M)$ be the set of all $k$-partitions of $M$. 

\begin{definition}[Proportionality]
    Given an instance $\cI = (M, N, \bc)$, for any $i \in N$ the proportional share of agent $i$ is $\frac{1}{n} \cdot c_i(M)$.
    An allocation is said to satisfy proportionality (PROP) if for every agent $i\in N$, $c_i(X_i)\le \frac{1}{n} \cdot c_i(M)$.
\end{definition}

\begin{definition}
    Given a set of items $M' \subseteq M$, a constant $k \leq n$, the $(M',k)$-maximin share of agent $i$ is defined as follows:
    \begin{align*}
        \MMS_i(M', k) = \min_{P \in \Pi_k(M')}\max_{j\in[k]} c_i(P_j).
    \end{align*}
\end{definition}

An $n$-partition of $M$ is called an MMS partition for agent $i$, if all the bundles in this partition have cost at most $\MMS_i(M,n)$.

\begin{definition}[$\alpha$-MMS]
    An allocation $\bX$ is said to satisfy $\alpha$-approximate maximin share for some $\alpha \geq 1$, if for every agent $i\in N$,
    \begin{equation*}
        c_i(X_i) \le \alpha \cdot \MMS_i(M,n).
    \end{equation*}
\end{definition}

When $\alpha = 1$, we call the allocation an MMS allocation.
For convenience, we simply write $\MMS_i$ when $M$ and $n$ are clear from the context.
Without loss of generality, for any agent $i\in N$ we normalize $c_i$ such that $\MMS_i = 1$. 
A bundle $X$ is called MMS-feasible for agent $i$ if $c_i(X) \leq 1$.
We may also simply say that agent $i$ likes bundle $X$ if it is MMS-feasible.
Given any partition $P = \{ P_1, \cdots P_k \}$ of items $M' \subseteq M$, we say $P$ is a $(\alpha, 1,\cdots, 1)$-partition for agent $i$ if there exists a bundle $P_j$ such that $c_i(P_j) \leq \alpha$ and $c_i(P_l) \leq 1$ for any $l \neq j$.

As a common observation of MMS, we have $c_i(M) \leq n\cdot \MMS_i = n$.\footnote{Otherwise, in any MMS partition $P$ of agent $i$ we have $c_i(P_j) < \frac{1}{n} \cdot c_i(M)$ for all $P_j \in P$, which leads to the contradiction that $\sum_{e\in M} c_i(e) < c_i(M)$.}
Furthermore, for any instance, we have $\MMS_i \geq c_i(e)$ for all $e\in M$ according to the definition of MMS.
Hence, we have that $c_i(e) \leq 1$ for all $i\in N$ and $e\in M$.

Next, we introduce the framework of all-but-one maximin share (AMMS).
We guarantee almost all agents ($n-1$ agents) an MMS allocation and the remaining one agent $\alpha$ times her MMS.

\begin{definition}[$\alpha$-AMMS]
    An allocation $\bX$ is $\alpha$-approximate all-but-one maximin share ($\alpha$-AMMS), if there exists an agent $i\in N$ with $c_i(X_i)\le \alpha\cdot \MMS_i$ and $c_j(X_j)\le  \MMS_j$ holds for every agent $j\in N\setminus \{i\}$.
\end{definition}



\begin{definition}[MMS-feasibility Graph]
    Given a partition $P= \{P_1,\dots,P_n\}$, we construct a bipartite graph $G = (N \cup P, E)$ in which we add an edge between agent $i \in N$ and bundle $P_j \in P$ if bundle $P_j$ is MMS-feasible for agent $i$, i.e., $E = \{(i, P_j ): i, j \in N, c_i(P_j ) \le 1\}$. 
    For any set of agents $S\subseteq N$, we use $L(S)$ to denote the set of neighbors of $S$.
\end{definition}

We remark that if $P$ is a partition of all items $M$, we have $|L(i)| \geq 1$ for every agent $i\in N$.
Otherwise, we have $\sum_{j=1}^n c_i(P_j)=c_i(M)>n$ which is a contradiction.
Furthermore, given an MMS-feasibility graph $G$, if there exists a perfect matching in $G$ then we can find an MMS allocation.
By Hall's Theorem, we have the following observation.

\begin{observation} \label{observation:perfectMatching}
    For any MMS-feasibility graph $G$, there is no perfect matching for $G$ if and only if there exists a subset $S \subseteq N$, such that $|S| > |L(S)|$. 
\end{observation} 

By Observation~\ref{observation:perfectMatching}, there must exist a subset $S\subseteq N$ such that $|S| > |L(S)|$, otherwise we can find a perfect matching that implies an MMS allocation.
We consider the maximum subset $S^*$ among them, i.e., $S^* = \argmax_{\{S\subseteq N: |S| > |L(S)|\}} |S|$. By definition, we have $N \setminus S^* \neq \emptyset$.

\begin{observation}\label{observation:perfect-matching-reduce}
    There exists a perfect matching between agents $N \setminus S^*$ and bundles $P \setminus L(S^*)$.
\end{observation}
\begin{proof}
    We show that for any subset $S' \subseteq N\setminus S^*$ we have $|S'| \leq |L(S')|$, which directly implies there exists a perfect matching between $N\setminus S^*$ and $P \setminus L(S^*)$, by Hall's Theorem.
    Assume otherwise that there exists a subset $S' \subseteq N\setminus S^*$ such that $|S'| > |L(S')|$.
    Then we have
    \begin{equation*}
        |L(S^*\cup S')| \leq |L(S^*)| + |L(S')| < |S^*| + |S'| = |S^*\cup S'|,
    \end{equation*}
    which is a contradiction towards that $S^*$ is the maximum subset with $|S^*| > |L(S^*)|$. See Figure~\ref{fig:example-mms-feasibility-graph} for an example.
\end{proof}

\vspace{-10pt}
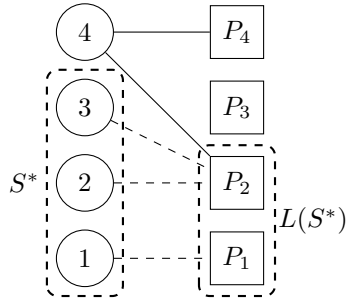
\begin{figure}[htbp]
    \centering
    \begin{tikzpicture}
    [round/.style={circle, draw, minimum size=7.5mm},
    square/.style={rectangle, draw, minimum size=7mm}]
    \foreach \x in {1, 2, 3, 4}
        \node (\x) at (0,\x) [round] {$\x$};
    \foreach \x in {1, 2, 3, 4}
        \node (P\x) at (2,\x) [square] {$P_\x$};
    \draw[dashed] (1) -- (P1);
    \draw[dashed] (2) -- (P2);
    \draw[dashed] (3) -- (P2);
    \draw (4) -- (P2);
    \draw (4) -- (P4);
    \draw[dashed, black, thick, rounded corners] (-0.5, 0.5) rectangle (0.5, 3.5);
    \draw[dashed, black, thick, rounded corners] (1.5, 0.5) rectangle (2.5, 2.5);
    \node at (-0.8, 2.0) {$S^*$};
    \node at (3, 1.5) {$L(S^*)$};
    \end{tikzpicture}
    \caption{An example of an MMS-feasibility graph with $n=4$ and  there exists a perfect matching between agent $\{4\}$ and bundles $\{P_3,P_4\}$.}
    \label{fig:example-mms-feasibility-graph}
\end{figure}


Based on the MMS-feasibility graph, we introduce valid reduction and reduced instances.
Given any instance $\cI = (M, N, \bc)$, a reduction is a procedure that allocates a set of bundles to a set of agents and outputs a new instance $\cI' = (M', N', \bc')$, where $M'$ and $N'$ are the unallocated items and agents that received nothing respectively.
Throughout this paper, each time we apply a reduction on an instance, we maintain that the reduction is ``valid'': for any agent $i\in N'$, we guarantee his/her proportional share does not increase, i.e., $c_i(M') \leq  |N'|$ for all $i\in N'$.

\begin{definition} [Reduced Instance]
    Given an $\cI = (M, N, \bc)$, a reduction is called a valid reduction if it outputs a new instance $\cI' = (M', N', \bc')$ such that $M' \subseteq M, N' \subseteq N$ and $c_i(M')\leq |N'|$ holds for all $i\in N'$. 
    We call such an instance $\cI'$ a reduced instance.
\end{definition}

\section{\texorpdfstring{$\frac{9}{8}$}{}-AMMS for Three Agents}

In this section, we consider instances with $n=3$ agents.
Given an instance $\cI = (M, N, \bc)$ with $n=3$, we let $P = \{P_1, P_2, P_3\}$ be an MMS partition for agent $3$.
We construct an MMS-feasibility graph $G = (N \cup P, E)$.
Note that we only have to consider the case that there is no perfect matching between $N$ and $P$.
Otherwise, there exists an MMS allocation.
Recall that $P$ is an MMS partition for agent $3$, all bundles are MMS-feasible for agent $3$, i.e., $L(\{3\}) = \{P_1, P_2, P_3\}$.
By Observation~\ref{observation:perfectMatching}, for agents $1,2$ we have $|L(\{1,2\})| = 1$.


When there is no perfect matching in the MMS-feasibility graph, we show that there exists a $\frac{9}{8}$-MMS allocation. 
Before introducing this main result, we first introduce the following technical lemma.

\begin{lemma} \label{lemma: MMS-valid-reduction}
    Given an allocation instance $(M,N,\bc)$, for any agent $i\in N$, if there exist two items $e_1, e_2$ with $c_i(e_1+e_2) \geq 1$, then $\MMS_i(M\setminus\{e_1,e_2\}, n - 1)\le \MMS_i(M, n)$.
\end{lemma}
\begin{proof}
     Recall that for any agent $i$, we have $\MMS_i(M,n) = 1$.
     Given an MMS partition $P = \{P_1,\dots,P_n\}$ for agent $i$, we have $c_i(P_j) \leq 1$ for all $P_j\in P$.
     Hence for items $e_1, e_2$ such that $c_i(e_1 + e_2) \geq 1$, either there exists a bundle $P^* \in P$ such that $P^* = \{e_1, e_2\}$ or there exist two bundles $P_1, P_2 \in P$ such that $e_1 \in P_1, e_2 \in P_2$ and $P_1 \neq P_2$.
     For the case that $P^* = \{e_1, e_2\}$, we have $P \setminus \{P^*\}$ is a partition of items $M\setminus \{e_1, e_2\}$ with $c_i(P_j) \leq 1 $ for all $P_j\in P \setminus \{P^*\}$.
     For the case that $e_1 \in P_1, e_2\in P_2$ such that $P_1 \neq P_2$, we have $c_i(P_1\cup P_2) - c_i(e_1 + e_2) \leq 1$.
     Hence $P' = (P_1\cup P_2\setminus\{e_1, e_2\}, P_3, P_4, \ldots, P_n)$ is a partition of items $M\setminus \{e_1, e_2\}$ while $c_i(P_j) \leq 1$ for all $P_j\in P'$. Recall that $\MMS_i$ is the minimum cost agent $i$ can guarantee if she divides all the items into $n$ bundles and picks a bundle at last, thus we have $\MMS_i(M\setminus\{e_1,e_2\}, n - 1)\le 1$.
\end{proof}

\begin{theorem}\label{thm:three-agents}
    For any instance with three agents, there exist $9/8$-AMMS allocations.
\end{theorem}

Before we prove Theorem~\ref{thm:three-agents}, we first introduce the \emph{divide-and-choose} and \emph{load-balancing} procedures that we will use in the analysis.
The \emph{divide-and-choose} procedure is described as follows.
Given a set of items $M'$, a divider $i$, and a chooser $j$, we let agent $i$ divide items $M'$ into two bundles $P_1, P_2$ and agent $j$ choose first.

\begin{observation} \label{observation:divide-and-choose}
    The divider receives a bundle with cost at most $\MMS_i (M',2)$ and the chooser receives a bundle with cost at most $\frac{1}{2} \cdot c_j(M')$.
\end{observation}
\begin{algorithm}
    \caption{Divide-and-Choose($M', i, j$)}
    \KwIn{A set of items $M'$, a divider $i$ and a chooser $j$.}
    
    Let agent $i$ divide all items into two bundles $P_1, P_2$ \;
    $X_j \gets \argmin_{B\in \{ P_1, P_2 \}}c_j(B)$\;
    $X_i \gets M'\setminus X_j$\;
    \KwOut{An allocation $\bX = (X_i, X_j)$.}
\end{algorithm}

Another well-studied procedure is called \emph{load-balancing}.
Given a set of items $M'$, an integer $n'$, and a cost function $c_i$, the algorithm follows the decreasing order of items and assigns each item to the bundle with minimum cost.

\begin{algorithm}[htbp]
    \caption{Load-Balancing($M', n', c_i$)}
    \label{alg:LBA}
    \KwIn{Item set $M'$, partition size $n'$, and cost function $c_i$ with $c_i(e_1) \geq c_i(e_2) \geq \cdots \geq c_i(e_{|M'|})$.}
    
    For all $k\in [n']$, let $P_k \gets \emptyset$ \;
    \For{$j = 1,2,\dots, |M'|$}{
        Let $k^* \gets \argmin_{k\in [n']} c_i(P_k)$\;
        Update $P_{k^*} \gets P_{k^*} + e_j$\;
    }
    \KwOut{A partition $P = \{P_1,\dots,P_{n'}\}$.}
\end{algorithm}

\begin{lemma}
    Given a partition $P$ returned by \emph{load-balancing}, if $c_i(M')\le n'$, for any bundle $P_j\in P$, for all item $e\in P_j$ we have $c_i(P_j - e) < 1$.
\end{lemma}
\begin{proof}
    Note that the algorithm allocates items with decreasing cost.
    To prove the lemma, it remains to consider the last item $e$ that is allocated to bundle $P_j$.
    Assume, for the sake of contradiction, that $c_i(P_j - e) \ge 1$.
    Following the algorithm, we must have $c_i(P_l) \ge 1$ for any other bundle $P_l \neq P_j$.
    Otherwise, the algorithm would not allocate item $e$ to bundle $P_j$.
    By summing the cost of all bundles we have $c_i(M') > n'$, which is a contradiction.
\end{proof}


Now we are ready to prove the existence of $\frac{9}{8}$-AMMS allocations.
We remark that we have an MMS-feasibility graph $G = (N \cup P, E)$ where $P = \{P_1,P_2,P_3\}$ is an MMS partition for agent $3$.
Moreover, we only consider the case that there is no perfect matching between $N$ and $P$, i.e., $|L(\{1,2\})| = 1$.
We assume w.l.o.g. that bundle $P_1$ is MMS-feasible for agents $1,2$.
We consider whether there exist two items $e_1, e_2$ in one of the bundles of $\{P_2, P_3\}$ and an agent $i\in \{1,2\}$ such that $c_i(e_1 + e_2) \geq 1$.
If there exist such two items $e_1, e_2$, then by Lemma~\ref{lemma: MMS-valid-reduction} we have $\MMS_i (M\setminus\{e_1,e_2\}, 2) \leq \MMS_i (M,3)$.
We show that by using such a reduction, we find an MMS allocation by calling \emph{divide-and-choose} procedure.
If there do not exist such two items, we show that by calling the well-studied \emph{load-balancing} procedure, we find a $\frac{9}{8}$-AMMS allocation. 
The following lemma considers the former case that there exists such two items $e_1, e_2$.
\begin{lemma} \label{lemma:MMS-reduction-3agents}
    If there exist two items $e_1,e_2$ in some bundle $P_j \in \{P_2, P_3\}$ with $c_i(e_1+e_2) \geq 1$ for some agent $i \in \{ 1, 2 \}$, there exists an MMS allocation.
\end{lemma}
\begin{proof}
    Assume w.l.o.g. that there exist two items $e_1,e_2$ in bundle $P_3$ with $c_1(e_1+e_2) \geq 1$.
    Recall that for agent 1, only bundle $P_1$ is MMS-feasible.
    By Lemma~\ref{lemma: MMS-valid-reduction}, we have
    \begin{equation*}
        \MMS_1(M\setminus P_3, 2)\le \MMS_1(M\setminus\{e_1,e_2\}, 2)\le \MMS_1(M, 3).
    \end{equation*}
    Furthermore, we have $c_2(M \setminus P_3) \leq 2$.
    We call the \emph{divide-and-choose} procedure on items $M\setminus P_3$ with agent $1$ being the divider and agent $2$ being the chooser.
    By Observation~\ref{observation:divide-and-choose}, it returns an allocation $(X_1, X_2)$ such that $c_1(X_1) \leq \MMS_1(M\setminus P_3, 2) \leq \MMS_1(M, 3)$ and $c_2(X_2) \leq \frac{1}{2} \cdot c_2(M\setminus P_3) \leq 1$.
    Recall that $P$ is an MMS partition for agent $3$.
    Then allocating bundle $P_3$ to agent $3$ leads to an MMS allocation.
\end{proof}

In the following, we consider the case that for any two items $e_1, e_2$ from the same bundle in $\{P_2, P_3\}$, we have $c_i(e_1+e_2) < 1$ for every $i\in \{1,2\}$.

\begin{lemma} \label{lemma:load-balancing-3agents}
    If for any two items $e_1,e_2$ in $P_2$ or $P_3$, we have $c_i(e_1+e_2) < 1$ for every $i \in \{ 1, 2 \}$, there exists a $\frac{9}{8}$-AMMS allocation.
\end{lemma}

\begin{proof}
Recall that for agents $1$ and $2$, only bundle $P_1$ is MMS-feasible for them.
We assume w.l.o.g. that $c_1(P_1) < c_1(P_2) \leq c_1(P_3)$. 
Note that when $c_1(P_2) \leq \frac{9}{8}$, we can get a $\frac{9}{8}$-AMMS allocation by allocating $P_2$ to agent $1$, $P_1$ to agent $2$ and $P_3$ to agent $3$.

Hence it remains to consider the case that $c_1(P_2) > \frac{9}{8}$, which implies that $c_1(P_3) \geq c_1(P_2) > \frac{9}{8}$. 
Let $M'=P_1\cup P_2$, we have $c_1(M') = c_1(M\setminus P_3) < \frac{15}{8}$. 
We call the \emph{load-balancing} procedure on items $M'$ and cost function $c_1$ to get a partition $\{B_1,B_2\}$.
We assume w.l.o.g. that $c_1(B_1) \geq c_1(B_2)$.
Again, if $c_1(B_1) \leq 9/8$, then we are done by letting agent $2$ pick the bundle she prefers in $\{B_1, B_2\}$, agent $3$ receive $P_3$ and agent $1$ receive the remaining items.

In the following, we assume that $c_1(B_1) > 9/8$ and show a contradiction.
We denote by $e$ the last item that is allocated to bundle $B_1$ (during \emph{load-balancing}) and let $x = c_1(e)$.
We use $B'_2 \subseteq B_2$ to denote the set of items that are allocated before $e$.
We have that $c_i(e') \geq c_i(e)$ for all $e'\in B_2$ following the \emph{load-balancing} procedure.

\begin{claim} \label{claim:B'_2-at-least-2x}
    We have $c_1(B'_2)\geq 2x$.
\end{claim}
\begin{proof}
    We consider two cases according to the size of $B_1$.
    When $|B_1| \geq 3$, we have $c_1(B_1 - e) \geq 2x$ since there exist at least two items in $B_1$ that are allocated before $e$.
    Following the \emph{load-balancing} procedure, we have $ c_1(B'_2) \geq c_1(B_1-e) \geq 2x$.

    Next, we consider the case that $|B_1| < 3$.
    We must have $|B_1| = 2$.
    Otherwise we have $|B_1| = 1$ and $c_1(e) = c_1(B_1) > 9/8$, which contradicts that $c_i(e) \leq 1$.
    Let $e'$ be the other item in $B_1$, we have $c_1(e + e') \geq 1$.
    We show that $|B'_2| \geq 2$, which directly implies that $c_1(B'_2) \geq 2x$ since all items in $B'_2$ are allocated before $e$.
    Suppose that $|B'_2| = 1$ and $e''$ is the item in $B'_2$.
    Then we have $c_1(e'') \geq c_1(e')$ since the algorithm allocates item $e$ to bundle $B_1$.
    Recall that we have $c_1(e + e') \geq 1$, which leads to that any two items in $\{e, e', e''\}$ have a total cost larger than $1$ for agent $1$.
    According to the Pigeonhole Principle, at least one of the bundles $\{P_1, P_2\}$ will contain two items in $\{e, e', e''\}$. Recall that we have $c_1(P_1)\le 1$, thus bundle $P_2$ must contain two of items in $\{e, e', e''\}$ which contradicts the assumption that for any two items $e_1, e_2$ from the same bundle in ${P_2, P_3}$, we have $c_i(e_1 + e_2) < 1$ for both $i \in \{1, 2\}$.
\end{proof}

According to the \emph{load-balancing} procedure, we have
\begin{equation*}
    c_1(B'_2) \geq c_1(B_1 - e) > \frac{9}{8} - x.
\end{equation*}
By Claim~\ref{claim:B'_2-at-least-2x}, we have $c_1(B'_2) \geq 2x$.
Combining the two bounds we have $3 \cdot c_1(B'_2) > \frac{18}{8}$, implying $c_1(B'_2) > \frac{3}{4}$.
Thus we have
\begin{equation*}
    c_1(M') = c_1(B_1) + c_1(B_2) > \frac{9}{8} + \frac{3}{4} = \frac{15}{8},
\end{equation*}
which contradicts that $c_1(M') = c_1(M\setminus P_3) < \frac{15}{8}$.
  
Combining the above analysis, there exists a $\frac{9}{8}$-AMMS allocation for three agents.
\end{proof}

For completeness, we summarize our complete algorithm in Algorithm~\ref{alg:three-agents}.
We further argue that our analysis is tight, by providing an example where Algorithm~\ref{alg:three-agents} returns an allocation such that $c_1(B_1) = \frac{9}{8}$.

\begin{algorithm}[htb]
    \caption{$\frac{9}{8}$-AMMS-Three-Agents$(M, N, \bc)$}\label{alg:three-agents}
    \KwIn{An instance $\cI = (M, N, \bc)$ with $|N| = 3$}
    
    Let $P=  \{P_1, P_2, P_3\}$ be an MMS partition for agent $3$.\\
    \eIf {Agents 1 and 2 like different bundles in $P$} {
        Allocate agents 1 and 2 bundles that they like respectively. \\
        Give the remaining bundle to agent 3.\\
    }
    {
        Assume w.l.o.g. $c_1(P_1) \leq 1$ and $c_2(P_1) \leq 1$.\\
        \tcp{Part 1:MMS reduction}
        \eIf{$\exists i\in \{1,2\}, j\in \{2,3\}$ and $e_1, e_2 \in P_j$, such that $c_i(e_1+e_2) \geq 1$}{
            $X_3 \gets P_j$; $M'\gets M\setminus X_3$; $N'\gets N\setminus \{3\}$\;
            $(X_1,X_2)\gets$ Divide-and-Choose($M', i, N' \setminus \{ i \}$)\;
        }
        {
            $X_3\gets \argmax_{P\in \{P_2,P_3\}} c_1(P)$.\\
            $M' \gets M\setminus X_3$.\\
        }
        \tcp{Part 2:Load-Balancing partition for agent 1}
        $\{B_1, B_2\} \gets$ Load-Balancing($M',|N'|,c_1$)\;
        $X_2 \gets \argmin_{B\in \{B_1, B_2\} } c_2(B)$.\\
        $X_1\gets M'\setminus X_2$
    }
    \KwOut{An allocation $\bX = (X_1, X_2, X_3)$}
\end{algorithm}

\begin{example}
    Consider an instance of allocating a set of $8$ items $M = \{e_k\}_{k\in [8]}$ to three agents $\{1,2,3\}$, where the cost functions of agents are defined as follows.

    \begin{center}
    \begin{tabular}{ c|c|c|c|c|c|c|c|c } 
        $\qquad$ & $e_1$ & $e_2$ & $e_3$ & $e_4$ & $e_5$ & $e_6$ & $e_7$ & $e_8$ \\
        \hline
        $\mathbf{1}$ & $3/8$ & $3/8$ & $3/8$ & $3/8$ & $3/8$ & $1/4$ & $1/4$ & $5/8$\\ 
        \hline
        $\mathbf{2}$ & $3/8$ & $3/8$ & $3/8$ & $3/8$ & $3/8$ & $1/4$ & $1/4$ & $5/8$\\
        \hline
        $\mathbf{3}$ & $1/2$ & $1/2$ & $1/3$ & $1/3$ & $1/3$ & $1/3$ & $1/3$ & $1/3$\\
    \end{tabular}
    \end{center}

    It can be verified that $\MMS_i = 1$ for any $i\in N$.
    We have $P = \{P_1, P_2, P_3\}$ with $P_1 = \{e_1, e_2\}, P_2 = \{e_3, e_4, e_5\}, P_3 = \{e_6, e_7, e_8\}$ being an MMS partition for agent $3$.
    Note that $P_1$ is the only bundle that agents $1$ and $2$ likes and $c_1(P_1) \leq c_1(P_2) \leq c_1(P_3)$ follows.
    Since for any two items in $P_2, P_3$, the total cost is strictly less than $1$ for both agents $1$ and $2$, we let agents $3$ pick bundle $P_3$ and allocate the remaining items using the \emph{load-balancing} procedure.
    The \emph{load-balancing} procedure returns bundles $B_1, B_2$ such that $c_i(B_1) = \frac{9}{8}$ and $c_i(B_2) = \frac{3}{4}$ for $i\in \{1,2\}$.
    Since agent $2$ picks first, agent $1$ receives the bundle $B_1$ which costs $\frac{9}{8}$ for her.
\end{example}

\section{\texorpdfstring{$\frac{4}{3}$}{}-AMMS for Four Agents}

In this section, we show the existence of $\frac{4}{3}$-AMMS allocations for $n=4$ agents.
Given an instance $\cI = (N, M, \bc)$ where $N = \{1,2,3,4\}$, let $P = \{P_1, P_2, P_3, P_4\}$ be an MMS partition for agent $4$ and we construct an MMS-feasibility graph $G = (N \cup P, E)$.
%
If there exists a perfect matching in the graph $G$, then we can find an MMS allocation.
Hence in the following, we only consider the case that there is no perfect matching in the MMS-feasibility graph $G$.
We show that in such a case, there exists a $\frac{4}{3}$-AMMS allocation.
For ease of presentation, in the subsequent MMS-feasibility graphs in this section, we omit the edges connected to agent $4$.


\begin{theorem} \label{theorem:fourAgents}
    For any instance $\cI = (M, N, \mathbf{c})$ with four agents, there exists a $\frac{4}{3}$-AMMS allocation.
\end{theorem}

Since there is no perfect matching in the graph $G$, there must exist a subset $S\subseteq N$ such that $|S|>|L(S)|$ following Observation~\ref{observation:perfectMatching}.
We use $S^*$ to denote the one with maximum size.
Note that $|S^*| \neq 1$ since each agent has at least one MMS-feasible bundle in any $n$-partition of $M$, i.e., $|L(\{i\})| \geq 1$ for all $i\in N$.
On the other hand, we have $|S^*| \neq 4$ since $P$ is an MMS partition of agent $4$, which means that $L(\{4\}) = \{P_1, P_2, P_3, P_4\}$.
Hence to show Theorem~\ref{theorem:fourAgents}, it suffices to consider $|S^*| = 2$ or $|S^*| = 3$.
We remark that by Observation~\ref{observation:perfect-matching-reduce}, there exists a perfect matching between $N\setminus S^*$ and $P\setminus L(S^*)$.
Hence, for both cases, it remains to only consider the agents in $S^*$, since we can allocate each agent in $N\setminus S^*$ an MMS-feasible bundle.
We show that such an operation is a valid reduction, by the following observation.

\begin{observation}  \label{observation:Valid-Reduction}
    If $S^* \subseteq N$ is the maximum subset such that $|S^*| > |L(S^*)|$, then we can do a valid reduction on $\cI$ by allocating each agent in $N\setminus S^*$ an MMS-feasible bundle.
\end{observation} 

\begin{proof}
By Observation~\ref{observation:perfect-matching-reduce}, there exists a perfect matching $\mathcal{M}$ between agents $N\setminus S^*$ and bundles $P\setminus L(S^*)$.
Let $\cI' = (M', N', \bc')$ be the instance after removing agents and bundles according to $\mathcal{M}$. 
Let $k' = |N \setminus S^*|$ and use $P_{\mathcal{M}} = \{P'_1,\dots,P'_{k'}\}$ to denote the bundles matched in $\mathcal{M}$. 
For any agent $i\in N'$ we have
    \begin{equation*}
        c_i(M') = c_i(M)- \sum_{j\in [k']} c_i(P'_j) <n - k' =|N'| .
    \end{equation*}
    The first inequality holds because $P_{\mathcal{M}}\subseteq P \setminus L(S^*)$ and for any $i \in N'$ and any bundle $P'_j \in P_{\mathcal{M}}$, we have $c_i(P'_j) > 1$.
\end{proof}


We first show that when $|S^*| = 2$, there exists a $\frac{4}{3}$-AMMS allocation.
Assume w.l.o.g. that $S^* = \{1,2\}$ and $L(S^*) = \{P_1\}$ (see Figure~\ref{fig:mmsPartitionGraph|S|=2|N(S)|=1}).
Note that following Observation~\ref{observation:Valid-Reduction}, it remains to consider the reduced instance $\cI' = (M', N', \bc')$ while $N' = \{1,2\}$\footnote{We simply use $M'$ to denote the remaining items after we allocated bundles to agents $3,4$. The only property of $M'$ we will use is that $c_i(M') \leq 2$ for $i\in \{1,2\}$, which is maintained by the valid reduction.}.
%
%
We show that by calling the \emph{load-balancing} procedure, we can compute an allocation such that $c_1(X_1) \leq \frac{4}{3}$ and $c_2(X_2) \leq 1$.

\begin{figure}[htbp]
    \centering
    \begin{tikzpicture}
    [round/.style={circle, draw, minimum size=7.5mm},
    square/.style={rectangle, draw, minimum size=7mm}]
    \foreach \x in {1, 2, 3, 4}
        \node (\x) at (0,\x) [round] {$\x$};
    \foreach \x in {1, 2, 3, 4}
        \node (P\x) at (2,\x) [square] {$P_\x$};
    \draw (2) -- (P1);
    \draw (1) -- (P1);
    \draw[dashed, black, thick, rounded corners] (-0.5, 0.5) rectangle (0.5, 2.5);
    \draw[dashed, black, thick, rounded corners] (1.5, 0.5) rectangle (2.5, 1.5);
    \node at (-0.8, 1.5) {$S^*$};
    \node at (3, 1) {$L(S^*)$};
    \end{tikzpicture}
    \caption{MMS-feasibility graph with $|S^*| = 2, |L(S^*)| = 1$.}
    \label{fig:mmsPartitionGraph|S|=2|N(S)|=1}
\end{figure}
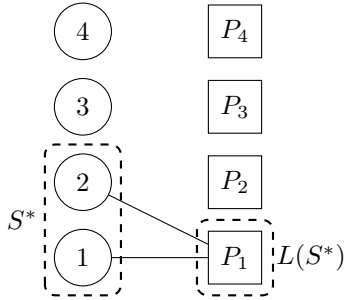

\begin{lemma}\label{lemma:|S|=2-load-balancing}
    Given a set of items $M'$ and an agent $i$ with $c_i(M')\leq 2$, the \emph{load-balancing} procedure returns a $(\frac{4}{3}, 1)$-partition for agent $i$.
\end{lemma}

\begin{proof}
    Let $\{B_1, B_2\}$ be the partition returned by the \emph{load-balancing} procedure.
    We assume w.l.o.g. that $c_i(B_1) \geq c_i(B_2)$.
    Note that $c_i(B_2) \leq \frac{1}{2} \cdot c_i(M') \leq 1$.
    It remains to show that $c_i(B_1) \leq \frac{4}{3}$.
    Suppose by contradiction that $c_i(B_1) > \frac{4}{3}$.
    Recall that $c_i(e) \leq 1$ for all $e\in M'$, which implies that $|B_1| \geq 2$.
    Let $e$ be the last item that is allocated to $B_1$.
    According to the \emph{load-balancing} procedure, we have
    \begin{equation*}
        c_i(B_1 - e) \leq c_i(B_2) = c_i(M') - c_i(B_1) < \frac{2}{3}.
    \end{equation*}
    Recall that $c_i(B_1) > \frac{4}{3}$, we have $c_i(e) > \frac{2}{3}$, which is a contradiction since the \emph{load-balancing} procedure allocates items following the decreasing order of costs.
    In other words, the cost of item $e$ cannot be larger than the cost of any item in $B_1 - e$.
\end{proof}

Combining Observation~\ref{observation:Valid-Reduction} and Lemma~\ref{lemma:|S|=2-load-balancing}, we have the following lemma.
\begin{lemma} \label{lemma:|S|=2}
    When $|S^*| = 2$, there exists a $\frac{4}{3}$-AMMS allocation.
\end{lemma} 
\begin{proof}
    Following Observation~\ref{observation:perfect-matching-reduce}, there exists a perfect matching between agents $\{3,4\}$ and bundles $\{ P_2, P_3, P_4 \}$.
    After allocating agents $3,4$ their MMS-feasible bundles, it remains to consider the allocation of remaining items $M'$ and agents $1,2$ with $c_i(M') < 2$ for $i\in \{1,2\}$.
    By Lemma~\ref{lemma:|S|=2-load-balancing}, calling the \emph{load-balancing} procedure on $c_1$ returns a $(\frac{4}{3}, 1)$-partition for agent $1$.
    Since $c_2(M') < 2$, there exists a bundle in the partition that is MMS-feasible for agent $2$.
    We allocate the MMS-feasible bundle to agent $2$ and the remaining bundle to agent $1$, which costs at most $\frac{4}{3}$ to her.
\end{proof}

To prove Theorem~\ref{theorem:fourAgents}, it remains to consider the case that $|S^*| = 3$.
It must hold that $S^* = \{1,2,3\}$ since all bundles in $P$ are MMS-feasible for agent $4$.
We complete the proof by showing the following two lemmas. Due to page limitations, we defer their proofs to Appendix \ref{appendix:missing-proofs}.

\begin{lemma} \label{lemma:|S|=3|N(S)|=2}
    When $|S^*| = 3$ and $|L(S^*)| = 2$, there exists a $\frac{4}{3}$-AMMS allocation.
\end{lemma}

\begin{lemma} \label{lemma:|S|=3|N(S)|=1}
    When $|S^*| = 3$ and $|L(S^*)| = 1$, there exists a $\frac{4}{3}$-AMMS allocation.
\end{lemma}


By Observation~\ref{observation:Valid-Reduction}, it remains to consider the reduced instance $\cI' = (M', N', \bc')$ where $N' = \{1,2,3\}$.
%
In the following, we propose a procedure for the reduced instance $\cI'$.
We show that when $\cI'$ ensures some property, there exists a $\frac{4}{3}$-AMMS allocation for the original instance $\cI$.

\begin{lemma} \label{lemma:4/3-partition}
    Given a reduced instance $\cI' = (M', N', \bc')$ with $|N'| = 3$,
    if there exists $T \subseteq M'$ and $i \in N'$ s.t. $c_i(T) \leq 1$ and $c_i(M' \setminus T) \leq 2$, then we can compute a $(\frac{4}{3}, 1, 1)$-partition for agent $i$.
\end{lemma}

\begin{proof}
    We call the \emph{load-balancing} procedure on items $M'\setminus T$ with respect to $c_i$.
    By Lemma~\ref{lemma:|S|=2-load-balancing}, the returned partition $\{B_1, B_2\}$ is a $(\frac{4}{3}, 1)$-partition for agent $i$.
    Hence for agent $i$, the partition $\{ T, B_1, B_2 \}$ is a $(\frac{4}{3}, 1, 1)$-partition.
%
\end{proof}



\begin{lemma}\label{lemma:4/3MMSforThreeAgents}
    Given a reduced instance $\cI' = (M', N', \bc')$ with $|N'| = 3$, if there exists a $(\frac{4}{3}, 1, 1)$-partition for some agent $i\in N'$, then we can compute a $\frac{4}{3}$-AMMS allocation for the original instance $\cI = (M, N, \bc)$.
\end{lemma}
\begin{proof}
We remark that agent $4$ receives a bundle that is MMS-feasible only for her, hence it remains to consider the reduced instance $\cI'$.
Let $\{B'_1, B'_2, B'_3\}$ be a $(\frac{4}{3}, 1, 1)$-partition (of items $M'$) for agent $1$.
We consider which bundles agents $2$ and $3$ like. 
If there exists a perfect matching between agents $\{2, 3\}$ and bundles $\{B'_1, B'_2, B'_3\}$, we allocate agents $2$ and $3$ each a bundle following the matching and agent $1$ receive the remaining bundle.
The bundle agent $1$ received costs at most $\frac{4}{3}$ to her.

Consider otherwise that there is only one bundle being MMS-feasible for agents $2$ and $3$. 
There must exist a bundle that is MMS-feasible for agent $1$ but not for agents $2$ and $3$. 
After assigning such a bundle to agent $1$, the total cost of the remaining items is less than $2$ for agents $2$ and $3$.
We call the \emph{load-balancing} procedure on the remaining items w.r.t. $c_2$.
By Lemma~\ref{lemma:|S|=2-load-balancing}, the returned partition is a $(\frac{4}{3},1)$-partition for agent $2$.
By letting agent $3$ pick the bundle she prefers, agent $2$ receives the remaining bundle which costs at most $4/3$.
\end{proof}

We conclude the above procedure in Algorithm~\ref{alg:AMMS-Three-Agents}.

\begin{algorithm}[H]
    \caption{$\frac{4}{3}$-AMMS-Three-Agents($\cI', B'_1, B'_2, B'_3$)}
    \label{alg:AMMS-Three-Agents}
    \KwIn{Reduced instance $\cI' = (M', N', \bc')$, a partition $\{B'_1, B'_2, B'_3\}$ that is a $(\frac{4}{3}, 1, 1)$-partition for agent $1$.}
        \eIf{There exists a perfect matching between agents 2, 3 and bundles $\{B'_1, B'_2, B'_3 \}$}{
            Allocate the bundle to agents 1 and 2 respectively based on the perfect matching. Agent $1$ receives the remaining one.
        }
        {
            Let agent $1$ pick the bundle that she likes while agents $2$ and $3$ dislike.\\
            $M' \gets M' \setminus X_1$.\\
            $(X_2, X_3) \gets$ Divide-and-Choose$(M', 2, 3)$.
        }
    \KwOut{$(X_1, X_2, X_3)$}
\end{algorithm}

Given Lemma~\ref{lemma:4/3-partition} and~\ref{lemma:4/3MMSforThreeAgents}, in the following of this section we mainly focus on computing such a $(\frac{4}{3}, 1, 1)$-partition for some agent.
We first consider the case that $|S^*| = 3$ and $|L(S^*)| = 2$, in which there are two bundles, say $P_3$ and $P_4$, that are MMS-feasible only for agent $4$.

\begin{replemma} {lemma:|S|=3|N(S)|=2}
    When $|S^*| = 3$ and $|L(S^*)| = 2$, there exists a $\frac{4}{3}$-AMMS allocation.
\end{replemma}

\begin{proof}
    We first consider the case that each agent in $S^*$ only has one MMS-feasible bundle.
    Since $|S^*| = 3, |L(S^*)| = 2$, there must exist two agents that like the same bundle, say that agents $1,2$ like bundle $P_1$.
    As Figure \ref{fig:|S|=3|N(S)|=2A} shows, we can allocate bundle $P_2$ to agent $3$ and bundle $P_4$ to agent $4$, which guarantees MMS for both of them.
    By Observation~\ref{observation:Valid-Reduction}, it remains to consider the allocation of items $P_1 \cup P_3$ to agents $1,2$ with $c_i(P_1 \cup P_3) < 2$ for both $i\in \{1,2\}$.
    By Lemma~\ref{lemma:|S|=2-load-balancing}, we can find a $(\frac{4}{3}, 1)$-partition on items $P_1 \cup P_3$ for agents $1$.
    Letting agent $2$ pick first and agent $1$ receive the remaining one results in a $\frac{4}{3}$-AMMS allocation.

\begin{figure}[htbp]
\centering
\begin{subfigure}{.25\textwidth}
  \centering
  \begin{tikzpicture} 
        [round/.style={circle, draw, minimum size=7.5mm},
        square/.style={rectangle, draw, minimum size=7mm}]
        \foreach \x in {1, 2, 3, 4}
           \node (\x) at (0,\x) [round] {$\x$};
        \foreach \x in {1, 2, 3, 4}
            \node (P\x) at (2,\x) [square] {$P_\x$};
        \draw (1) -- (P1);
        \draw (2) -- (P1);
        \draw (3) -- (P2);
        \draw[dashed, black, thick, rounded corners] (-0.5, 0.5) rectangle (0.5, 3.5);
        \draw[dashed, black, thick, rounded corners] (1.5, 0.5) rectangle (2.5, 2.5);
        \node at (-0.8, 2) {$S^*$};
        \node at (3, 1.5) {$L(S^*)$};
        \end{tikzpicture}
  \caption{MMS feasibility graph with $|S^*| = 3, |L(S^*)| = 2$ while every agent likes only one bundle.}
  \label{fig:|S|=3|N(S)|=2A}
\end{subfigure}\qquad \quad
\begin{subfigure}{.25\textwidth}
  \centering
  \begin{tikzpicture} 
        [round/.style={circle, draw, minimum size=7.5mm},
        square/.style={rectangle, draw, minimum size=7mm}]
        \foreach \x in {1, 2, 3, 4}
           \node (\x) at (0,\x) [round] {$\x$};
        \foreach \x in {1, 2, 3, 4}
            \node (P\x) at (2,\x) [square] {$P_\x$};
        \draw (1) -- (P1);
        \draw (2) -- (P1);
        \draw (3) -- (P2);
        \draw (1) -- (P2);
        \draw[dashed, black, thick, rounded corners] (-0.5, 0.5) rectangle (0.5, 3.5);
        \draw[dashed, black, thick, rounded corners] (1.5, 0.5) rectangle (2.5, 2.5);
        \node at (-0.8, 2) {$S^*$};
        \node at (3, 1.5) {$L(S^*)$};
        \end{tikzpicture}
  \caption{MMS feasibility graph with $|S^*| = 3, |L(S^*)| = 2$ while there exists an agent likes more than one bundle.}
  \label{fig:|S|=3|N(S)|=2B}
\end{subfigure}\qquad \quad
\begin{subfigure}{.25\textwidth}
    \centering
    \begin{tikzpicture}
    [round/.style={circle, draw, minimum size=7.5mm},
    square/.style={rectangle, draw, minimum size=7mm}]
    \foreach \x in {1, 2, 3, 4}
        \node (\x) at (0,\x) [round] {$\x$};
    \foreach \x in {1, 2, 3, 4}
        \node (P\x) at (2,\x) [square] {$P_\x$};
    \draw (2) -- (P1);
    \draw (1) -- (P1);
    \draw (3) -- (P1);
    \draw[dashed, black, thick, rounded corners] (-0.5, 0.5) rectangle (0.5, 3.5);
    \draw[dashed, black, thick, rounded corners] (1.5, 0.5) rectangle (2.5, 1.5);
    \node at (-0.8, 2) {$S^*$};
    \node at (3, 1) {$L(S^*)$};
    \end{tikzpicture}
    \caption{MMS-feasibility graph with $|S^*| = 3, |L(S^*)| = 1$, for which all of agents $\{1,2,3\}$ only like bundle $P_1$.}
    \label{fig:mmsPartitionGraph|S|=3|N(S)|=1}
\end{subfigure}
\caption{MMS feasibility graphs with $|S^*| = 3, |L(S^*)| \leq 2$.}
\end{figure}
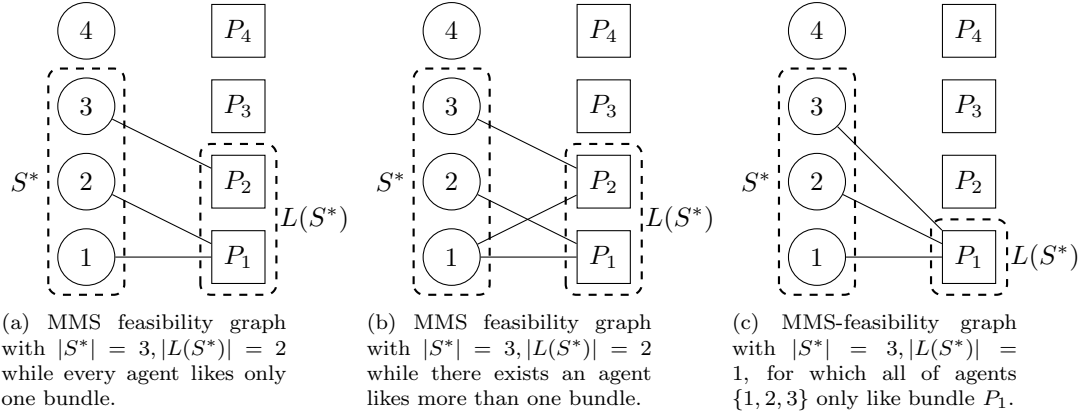

    Next, we move to the case that there exists an agent that likes two bundles (see Figure~\ref{fig:|S|=3|N(S)|=2B}). 
    Without loss of generality, we assume agent $1$ likes two bundles $P_1, P_2$ and $c_1(P_1)\leq c_1(P_2) \leq c_1(P_3) \leq c_1(P_4)$. 
    By allocating bundle $P_4$ to agent $4$, it remains to consider $\cI' = (M', N', \bc')$ with $N' = \{1,2,3\}, M' = M\setminus P_4$ since $c_i(M \setminus P_4) < 3$ for all $i \in \{ 1, 2, 3 \}$. 
    Notice that 
    \begin{equation*}
        c_1(P_1) + c_1(P_3) \leq c_1(P_2) + c_1(P_4) \leq \frac{1}{2} \cdot c_i(M) \leq 2
    \end{equation*}
    and $c_1(P_2) < 1$. 
    By Lemma~\ref{lemma:4/3-partition}, we can compute a $(\frac{4}{3}, 1, 1)$-partition for agent 1.
    Based on the partition, we can find a $\frac{4}{3}$-AMMS allocation by Lemma~\ref{lemma:4/3MMSforThreeAgents}.
\end{proof}

It remains to consider the case that all agents $1, 2, 3$ like the same bundle, say $P_1$ (See Figure~\ref{fig:mmsPartitionGraph|S|=3|N(S)|=1}).
We remark that this case is the most complicated and needs a more careful analysis.
One of the differences compared with the previous analysis is that we not only need the valid reduction property, i.e.,  $c_i(M') \leq 3$ for all $i\in N'$; but also need to specify which bundle is allocated to agent $4$.
To specify such a bundle, we use the \emph{atomic bundle} framework that was first introduced by Feige $et$ $al$.~\cite{conf/wine/FeigeST21}.
More specifically, the \emph{atomic bundle} is the intersection of two bundles, each of them from a different agent's MMS partitions.
Note that in this case, agents $1,2,3$ are symmetric.
In the following, we fix agent $1$ and let $B = \{B_1, B_2, B_3, B_4\}$ be an MMS partition of agent $1$.
We take the intersection of bundles in $B$ and bundles in $P$, i.e., $b_{ij} = B_i \cap P_j$ for all $i, j \in \{ 1, 2, 3, 4\}$. 
We show the atomic bundles as a matrix in Figure~\ref{fig:atomicBundle}.

\begin{figure}[htbp]
    \centering
    \begin{subfigure}{.25\textwidth}
    \centering
    \begin{tikzpicture}
    [square/.style={rectangle, draw, minimum size=7mm}]
    \matrix [nodes={draw, minimum size=8mm}] {
        \node {$b_{11}$}; & \node{$b_{12}$}; & \node{$b_{13}$}; & \node{$b_{14}$};\\
        \node {$b_{21}$}; & \node{$b_{22}$}; & \node{$b_{23}$}; & \node{$b_{24}$};\\
        \node {$b_{31}$}; & \node{$b_{32}$}; & \node{$b_{33}$}; & \node{$b_{34}$};\\
        \node {$b_{41}$}; & \node{$b_{42}$}; & \node{$b_{43}$}; & \node{$b_{44}$};\\
    };
    \node at (-1.2, 2) {$P_1$};
    \node at (-1.2, 2.5) {$\leq 1$}; 
    \node at (-0.4, 2) {$P_2$};
    \node at (-0.4, 2.5) {$>1$};
    \node at (0.4, 2) {$P_3$};
    \node at (0.4, 2.5) {$>1$};
    \node at (1.2, 2) {$P_4$};
    \node at (1.2, 2.5) {$>1$};
    \node at (-2, 1.1) {$B_1$};
    \node at (-2, 0.4) {$B_2$};
    \node at (-2, -0.3) {$B_3$};
    \node at (-2, -1.1) {$B_4$};
    \end{tikzpicture}
    \caption{The atomic bundles as a matrix.}
    \label{fig:atomicBundle}
    \end{subfigure}\qquad
    \begin{subfigure}{.25\textwidth}
    \centering
    \begin{tikzpicture}
    \filldraw[fill = gray!20, draw = none] (-1.6,1.6) rectangle (1.6,-0.8);
    \filldraw[fill = blue!20, draw = none] (-0.8,-0.8) rectangle (0,-1.6);
    [square/.style={rectangle, draw, minimum size=7mm}]
    \matrix [nodes={draw, minimum size=8mm}] {
        \node {$b_{11}$}; & \node{$b_{12}$}; & \node{$b_{13}$}; & \node{$b_{14}$};\\
        \node {$b_{21}$}; & \node{$b_{22}$}; & \node{$b_{23}$}; & \node{$b_{24}$};\\
        \node {$b_{31}$}; & \node{$b_{32}$}; & \node{$b_{33}$}; & \node{$b_{34}$};\\
        \node {$b_{41}$}; & \node{$b_{42}$}; & \node{$b_{43}$}; & \node{$b_{44}$};\\
    };
    \node at (-1.2, 2) {$P_1$};
    \node at (-0.4, 2) {$P_2$};
    \node at (0.4, 2) {$P_3$};
    \node at (1.2, 2) {$P_4$};
    \node at (-2, 1.1) {$B_1$};
    \node at (-2, 0.4) {$B_2$};
    \node at (-2, -0.3) {$B_3$};
    \node at (-2, -1.1) {$B_4$};
    \end{tikzpicture}
    \caption{The gray bundles cost at most $3$ to agent $1$.}
    \label{fig:atomicBundle2}
\end{subfigure}\qquad
\begin{subfigure}{.25\textwidth}
    \centering
    \begin{tikzpicture}
    \filldraw[fill = gray!20, draw = none] (-1.6,1.6) rectangle (1.6,-0.8);
    \filldraw[fill = blue!20, draw = none] (-0.8,-0.8) rectangle (0,-1.6);
    \filldraw[fill = white, draw = none] (-0.8,1.6) rectangle (0,-0.8);
    [square/.style={rectangle, draw, minimum size=7mm}]
    \matrix [nodes={draw, minimum size=8mm}] {
        \node {$b_{11}$}; & \node{$b_{12}$}; & \node{$b_{13}$}; & \node{$b_{14}$};\\
        \node {$b_{21}$}; & \node{$b_{22}$}; & \node{$b_{23}$}; & \node{$b_{24}$};\\
        \node {$b_{31}$}; & \node{$b_{32}$}; & \node{$b_{33}$}; & \node{$b_{34}$};\\
        \node {$b_{41}$}; & \node{$b_{42}$}; & \node{$b_{43}$}; & \node{$b_{44}$};\\
    };
    \node at (-1.2, 2) {$P_1$};
    \node at (-0.4, 2) {$P_2$};
    \node at (-0.4, 2.5) {$>1$};
    \node at (0.4, 2) {$P_3$};
    \node at (1.2, 2) {$P_4$};
    \node at (-2, 1.1) {$B_1$};
    \node at (-2, 0.4) {$B_2$};
    \node at (-2, -0.3) {$B_3$};
    \node at (-2, -1.1) {$B_4$};
    \end{tikzpicture}
    \caption{The gray bundles cost at most $2$ to agent $1$}
    \label{fig:atomicBundle3}
\end{subfigure}
    \caption{Each atomic bundle is an intersection of a bundle in $B$ (an MMS partition of agent $1$) and a bundle in $P$ (an MMS partition of agent $4$). We remark that $P_1$ is the bundle that agent $1$ likes while she dislikes all other bundles $P_2, P_3, P_4$.}
\end{figure}

\begin{replemma} {lemma:|S|=3|N(S)|=1}
    When $|S^*| = 3$ and $|L(S^*)| = 1$, there exists a $\frac{4}{3}$-AMMS allocation.
\end{replemma}

\begin{proof}
    We first consider the case that there exists an atomic bundle $b_{ij}$, such that $c_1(P_j \setminus b_{ij}) > 1$ for some $i \in \{ 1, 2, 3, 4 \}, j \in \{2, 3, 4\}$.
    We show that in this case, Lemma~\ref{lemma:4/3-partition} can be applied to find a $\frac{4}{3}$-AMMS allocation.
    Note that $b_{ij}$ appears in $P_2, P_3, P_4$ symmetrically. We assume w.l.o.g. that $b_{42} \subseteq P_2$ is such an atomic bundle (see Figure~\ref{fig:atomicBundle2}).
    By assigning $P_2$ to agent $4$, it remains to consider the reduced instance $\cI' = (M', N', \bc')$ while $M' = M\setminus P_2$ and $N' = \{1,2,3\}$.
    Let $T = b_{41}\cup b_{43}\cup b_{44}$, we have $c_1(T) \leq 1$ since $T\subseteq B_4$ while $c_1(B_4) \leq 1$.
    Moreover, since $c_1(P_2 \setminus b_{42}) > 1$, we have
    \begin{equation*}
        c_1(M'\setminus T) = c_1(B_1 \cup B_2 \cup B_3 \setminus P_2) < 2.
    \end{equation*}
    See Figure~\ref{fig:atomicBundle3} for an illustration.
    By applying Lemma~\ref{lemma:4/3-partition}, we can compute a $(\frac{4}{3}, 1, 1)$-partition for agent 1.
    Based on the partition, we can find a $\frac{4}{3}$-AMMS allocation by Lemma~\ref{lemma:4/3MMSforThreeAgents}.

    It remains to consider that for any $b_{ij}$ such that $i\in \{1,2,3,4\}, j\in \{2,3,4\}$, we have $c_1(P_j \setminus b_{ij}) \leq 1$. 
    We assume w.l.o.g. that $c_1(P_1) \leq c_1(P_2) \leq c_1(P_3) \leq c_1(P_4)$. 
    By assigning $P_4$ to agent $4$, it remains to consider the reduced instance $\cI' = (M', N', \bc')$ while $M' = M\setminus P_4$ and $N' = \{1,2,3\}$.
    Again, we show that there exists a $(\frac{4}{3}, 1, 1)$-partition for agent $1$.
    Let $c_1(P_1) = \beta$. 
    We have 
    \begin{equation*}
        c_1(P_2) + c_1(P_3) = c_1(M) - c_1(P_1)-c_1(P_4) \leq \frac{2(c_1(M) - \beta)}{3} \leq \frac{2(4 - \beta)}{3},
    \end{equation*}
    where the first  inequality holds because $c_1(P_1) + c_1(P_4) \geq \frac{c_1(M) + 2\beta}{3}$ and the second inequality holds because $c_1(M) \leq 4$.
    Let $b_2$ and $b_3$ be the atomic bundle with the minimum cost in $P_2$ and $P_3$ respectively, i.e., $b_2 = \argmin \{c_1(b): b \in \{b_{12}, b_{22}, b_{32}, b_{42}\} \}$, and $b_3 = \argmin \{c_1(b): b \in \{b_{13}, b_{23}, b_{33}, b_{43}\} \}$.

    \begin{itemize}
        \item Suppose $\beta \leq \frac{4}{5}$.
        We have 
        $$c_1(b_2) + c_1(b_3) \leq \frac{1}{4} \cdot \left(c_1(P_2) + c_1(P_3)\right) \leq \frac{1}{4} \cdot \frac{2(4 - \beta)}{3} = \frac{4 - \beta}{6},$$
        which implies that $c_1(P_1) + c_1(b_2) + c_1(b_3) \leq \beta + \frac{4 - \beta}{6} = \frac{4 + 5\beta}{6} \leq \frac{4}{3}$.
        Recall that $c_1(P_2\setminus b_2) \leq 1$ and $c_1(P_3 \setminus b_3) \leq 1$.
        Hence $\{P_1 \cup b_2 \cup b_3, P_2 \setminus b_2, P_3 \setminus b_3\}$ is a $(\frac{4}{3}, 1, 1)$-partition for agent $1$.

        \item Suppose $\beta > \frac{4}{5}$. 
        We have
        \begin{equation*}
            c_1(b_3) \leq \frac{1}{4} \cdot c_1(P_3) \leq \frac{1}{4} \cdot \frac{4 - \beta - c_1(P_2)}{2}.
        \end{equation*}
        Note that $c_1(P_2) \leq c_1(P_3)$ while $c_1(P_2) + c_1(P_3) \leq  \frac{2(4 - \beta)}{3}$.
        Hence we have $c_1(P_2) \leq \frac{4 - \beta}{3}$ and
        \begin{equation*}
            c_1(P_2) + c_1(b_3) \leq \frac{4 - \beta + 7c_1(P_2)}{8} \leq \frac{5}{3} - \frac{5}{12} \beta \leq \frac{4}{3}.
        \end{equation*}
        Recall that $c_1(P_3 \setminus b_3) \leq 1$ and $c_1(P_1) < 1$. 
        Hence $\{P_1, P_2 \cup b_3, P_3\setminus b_3\}$ is a $(\frac{4}{3}, 1, 1)$-partition for agent $1$.
        \qedhere
    \end{itemize}
\end{proof}

For completeness, we summarize our complete algorithm as Algorithm~\ref{alg:four-agents} in Appendix~\ref{appendix:complete-algo}.

\section{\texorpdfstring{$\frac{(n+1)^2}{4n}$}{}-AMMS for a General Number of Agents}

In this section, we consider the instances with a general number of agents. Let $\gamma = \frac{(n+1)^2}{4n}$. 
Given any instance $\cI=(M, N, \bc)$, we show that $\gamma$-AMMS allocations are guaranteed to exist, by proposing Algorithm~\ref{alg:n-agents}.
The algorithm runs in iterations while during each iteration it either returns a $\gamma$-AMMS allocation (and terminates); or reduces the input instance (which is a reduced instance) to a new reduced instance with a smaller size.
We remark that the original instance $\cI = (M, N, \bc)$ can also be viewed as a reduced instance since $c_i(M) \leq n = |N|$ for all $i\in N$.

At the beginning of each iteration, given a reduced instance $\cI'= (M', N', \bc')$ with $|N'| = k$, we find a partition $P = \{P_1, \ldots, P_k\}$ of item $M'$ that is a $(\gamma, 1, \ldots, 1)$-partition for some agent $i \in N'$.
Specifically, we assume that $c_i(P_j)\le 1$ for $j\in [k-1]$ and $c_i(P_k)\le \gamma$.
Based on the partition $P$ we construct an MMS-feasibility graph $G'$ excluding agent $i$ and bundle $P_k$.
We remark that if there exists a perfect matching in $G'$, then we can directly find a $\gamma$-AMMS allocation, by allocating $P_k$ to agent $i$ and the others following the perfect matching.
If there is no perfect matching, we claim that we can do a valid reduction on $\cI'$.

\begin{algorithm}[htbp]
        \caption{$\gamma$-AMMS-$n$-Agents$(M,N,\bc)$} \label{alg:n-agents}
        \KwIn{An original instance $\cI= (M, N, \bc)$}

        $M'\gets M$; $N'\gets N$; $\bc'\gets \bc $\;
        Let $\cI' \gets (M',N',\bc')$ and $k\gets |N'|$\;

        \While{$|N'|\ne \emptyset$}{
            Pick an arbitrary agent $i \in N'$\;
            Find a $(\gamma,1,\dots,1)$-partition $P$ for agent $i$; \hfill \tcp{Use Algorithm~\ref{alg:MP} or Algorithm~\ref{alg:CBF}}
    
            Construct an MMS-feasibility graph $G'$ between $N'\setminus \{i\}$ and $P \setminus \{P_k\}$\;
            \If{there exists a perfect matching $\mathcal{M}$ in $G'$} {
                Allocate bundles to the corresponding agents in $\mathcal{M}$ and allocate $P_k$ to agent $i$\;
            \Return $(X_1,\dots,X_n)$\;
            }
            $(\hat{M},\hat{N},\hat{\bc}), \{X_i\}_{i\in N' \setminus \hat{N}} \gets$ Valid-Reduction($G,i$); \hfill \tcp{Algorithm~\ref{alg:reduced}}
            Update $M' \gets \hat{M}$, $N' \gets \hat{N}$, $\bc' \gets \hat{\bc}$\;
            Let $\cI'\gets (M',N',\bc')$ and $k\gets |N'|$\;
        }
        \KwOut{An allocation $\bX = (X_1, \dots, X_n)$}
\end{algorithm}

\begin{theorem} \label{thm:general-number-agents}
    For any instance with $n\geq 5$ agents, there exist $\gamma$-AMMS allocations.
\end{theorem}

To prove Theorem~\ref{thm:general-number-agents}, we first show that the partition $P = \{P_1, \ldots, P_k\}$ we output during each iteration, is a $(\gamma, 1, \ldots, 1)$-partition for some agent $i\in N'$.
We remark that at the beginning of any iteration, we have a reduced instance $\cI' = (M', N', \bc')$ as input.

\begin{lemma} \label{lemma:gamma-partition}
   Given a reduced instance $\cI' = (M',N',\bc')$, we can compute a partition $P= \{P_1,\dots, P_k\}$ that is a $(\gamma,1,\dots,1)$-partition for some agent $i\in N'$.
\end{lemma}

Given such a partition $P$, we construct an MMS-feasibility graph $G$.
We mainly focus on the subgraph $G'$ with agents $N' \setminus \{i\}$, the set of bundles $P \setminus \{P_k\}$, and the edges between them (see Figure~\ref{fig:modified-MMS-feasibility-graph}).
If there exists a perfect matching $\mathcal{M}$ in the graph $G'$, by allocating $P_k$ to agent $i$ and the other bundles following the matching, we have a $\gamma$-AMMS allocation.
It remains to consider that there is no perfect matching in $G'$, in which case we run Algorithm~\ref{alg:reduced} to reduce the instance.

    \begin{figure}[htbp]
        \centering
        \begin{subfigure}{0.45\textwidth}
        \centering
        \begin{tikzpicture}
        [round/.style={circle, draw, minimum size=7.5mm},
        square/.style={rectangle, draw, minimum size=7mm},
        dot/.style={rotate=90, font=\LARGE}]
        
        \filldraw[fill=gray!10, draw=black, dashed, thick, rounded corners] (-0.5,5.5) rectangle (2.5, 4.5);

        \foreach \x in {1,2}
            \node (\x) at (0,\x) [round] {$\x$};
        \node [dot] at (0, 3) {$\cdots$};
        \node (i) at (0, 5) [round] {$i$};
        \node (k-1) at (0, 4) [round] {};
        \node at (0, 4) {\scriptsize$k-1$};
    
        \foreach \x in {1,2}
            \node (P\x) at (2,\x) [square] {$P_\x$};
        \node [dot] at (2, 3) {$\cdots$};\
        \node (Pk-1) at (2, 4) [square] {};
        \node at (2,4) {\scriptsize$P_{k-1}$};
        \node (Pk) at (2, 5) [square] {$P_k$};
        
        \draw (1) -- (P2);
        \draw (2) -- (P2);
        \draw (k-1) -- (P2);
        \draw (k-1) -- (Pk-1);
        \draw[dashed] (i) -- (Pk);

        \draw[dashed, black, thick, rounded corners] (-0.5,0.5) rectangle (2.5, 4.5);
        \node at (3, 2.5) {$G'$};
        \end{tikzpicture}
        \caption{The MMS-feasibility graph $G'$ is a subgraph of $G$ that excludes the agent $i$, the bundle $P_k$ and all the edges connected to them, e.g., $G' = (N\setminus \{i\} \cup P\setminus \{P_k\}, E)$.
        The dashed edge between $i$ and $P_k$ in the figure denotes that $c_i(P_k) \leq \gamma$.}
        \label{fig:modified-MMS-feasibility-graph}   
        \end{subfigure}\qquad
        \begin{subfigure}{.45\textwidth}
        \centering
         \begin{tikzpicture}
        [round/.style={circle, draw, minimum size=7.5mm},
        square/.style={rectangle, draw, minimum size=7mm},
        dot/.style={rotate=90, font=\LARGE}]
    
        \filldraw[fill=gray!10, draw=black, dashed, thick, rounded corners] (-0.5,5.5) rectangle (2.5, 4.5);
            
        \foreach \x in {1,2}
            \node (\x) at (0,\x) [round] {$\x$};
        \node [dot] at (0, 3) {$\cdots$};
        \node (i) at (0, 5) [round] {$i$};
        \node (k-1) at (0, 4) [round] {};
        \node at (0, 4) {\scriptsize$k-1$};
        
        \foreach \x in {1,2}
            \node (P\x) at (2,\x) [square] {$P_\x$};
        \node [dot] at (2, 3) {$\cdots$};\
        \node (Pk-1) at (2, 4) [square] {};
        \node at (2,4) {\scriptsize$P_{k-1}$};
        \node (Pk) at (2, 5) [square] {$P_k$};
        
        \draw (1) -- (P2);
        \draw (2) -- (P2);
        \draw (k-1) -- (P2);
        \draw (k-1) -- (Pk-1);
        \draw[dashed] (i) -- (Pk);
    
        \draw[dashed, black, thick, rounded corners] (-0.5,0.5) rectangle (0.5, 2.5);
        \draw[dashed, black, thick, rounded corners] (1.5,1.5) rectangle (2.5, 2.5);
        \node at (-0.8, 2) {$S$};
        \node at (3, 2) {$L(S)$};
        \end{tikzpicture}
        \caption{An illustration for $S$ and $L(S)$. We can let agents in $N'\setminus \{1,2,i\}$ pick bundles following the perfect matching $\mathcal{M}$ and agent $i$ picks one from the remaining bundles in $P'\setminus L(S)$. It remains to consider a reduced instance with only agents $\{1,2\}$.}
        \label{fig:weak mms partition}
    \end{subfigure}
    \caption{Illustrations for subgraph $G'$ and $S$ with $L(S)$.}
    \end{figure}
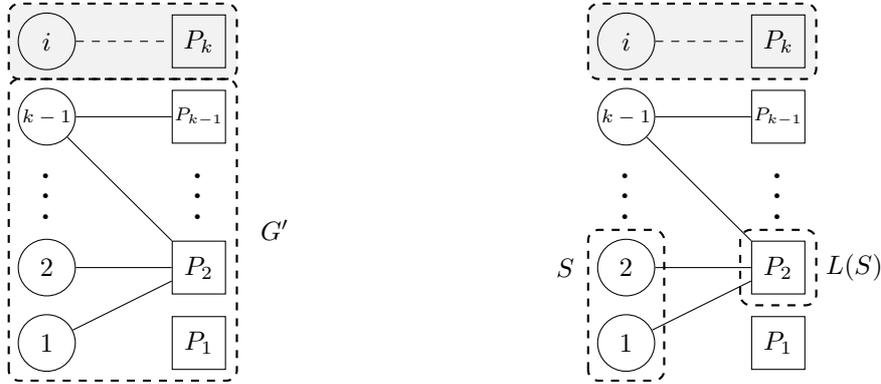

\begin{algorithm}[h]
        \caption{Valid-Reduction($G', i$)} \label{alg:reduced}
        \KwIn{An MMS-feasibility graph $G'=(N'\setminus \{i\} \cup P\setminus \{P_k\},E)$ and agent $i$} 
        
        Find a maximum subset $S\subseteq N'\setminus \{i\}$ in $G'$ such that $|S|>|L(S)|$\;
        \If{$S \neq N'\setminus \{i\}$}{
        Find a perfect matching $\mathcal{M}$ between $N' \setminus (S \cup \{i\})$ and $P \setminus (L(S) \cup \{P_k\}) $\;
        Let $N_{\mathcal{M}}$ be the set of agents matched in $\mathcal{M}$ and $P_{\mathcal{M}}$ be set of bundles matched in $\mathcal{M}$\;
        
        Allocate bundles to agents in $N_{\mathcal{M}}$ following the perfect matching $\mathcal{M}$.
        }
        Let agent $i$ pick any one of bundles in $P \setminus (L(S)\cup P_{\mathcal{M}} \cup \{P_k\})$ (use $X_i$ to denote this bundle)\;
        $\hat{M} \gets M' \setminus (P_{\mathcal{M}}\cup X_i)$\;
        $\hat{N} \gets N' \setminus (N_{\mathcal{M}}\cup i)$\;
        \KwOut{An instance $\hat{\cI} = (\hat{M}, \hat{N}, \hat{\bc})$ and $\{X_i\}_{i\in N' \setminus \hat{N}}$}
\end{algorithm}

\begin{lemma} \label{lemma:whether-mms-or-reduce}
    If there is no perfect matching in the graph $G'$, Algorithm~\ref{alg:reduced} can reduce the instance $\cI'$ to a new reduced instance $\hat{\cI}$ with a smaller size.
\end{lemma}

\begin{proof}
    Let $\hat{N} = N'\setminus \{i\}$  and $\hat{P} = P\setminus \{P_k\}$.\;
    If there exists no perfect matching between $\hat{N}$ and $\hat{P}$, according to Hall's theorem, we can find a maximum size subset $S\subseteq \hat{N}$ such that $|S|>|L(S)|$.
    By Observation~\ref{observation:perfect-matching-reduce}, if $\hat{N}\setminus S\ne \emptyset$, there exists a perfect matching $\mathcal{M}$ between $\hat{N}\setminus S$ and $\hat{P}\setminus L(S)$. We use $N_{\mathcal{M}}$ and $P_{\mathcal{M}}$ to denote the set of agents and set of bundles matched in $\mathcal{M}$ respectively.
    We allocate bundles in $P_{\mathcal{M}}$ to agents in $N_{\mathcal{M}}$ following the matching $\mathcal{M}$.
    Note that we have 
    \begin{equation*}
        |\hat{P}|- |L(S)|-|P_{\mathcal{M}}|=|\hat{P}|- |L(S)|-|\hat{N}\setminus S|=|S|-|L(S)| \ge 1,
    \end{equation*}
   which implies that we can allocate one of the bundles in $\hat{P}\setminus (L(S)\cup P_{\mathcal{M}})$ to agent $i$. 
    Recall that all bundles in $\hat{P}$ are MMS-feasible for agent $i$, and the bundle received by agent $i$ is MMS-feasible for her and not MMS-feasible for all agents in set $S$. 

    Since the bundles allocated in Algorithm~\ref{alg:reduced} are not MMS-feasible for agents in $S$, the proportional share of agents in $S$ will not increase. 
    On the other hand, any agent in $N' \setminus S$ receives a bundle in Algorithm~\ref{alg:reduced} while the bundle costs at most $1$ for her.
    Hence it remains to only consider the new reduced instance with only agents in $S$.
    Moreover, the new reduced instance has a strictly smaller size since $|N'|-|S|\ge 1$.
\end{proof}

\begin{proofof}{Theorem~\ref{thm:general-number-agents}} 
Following Lemma~\ref{lemma:gamma-partition}, given any reduced instance $\cI'$ with $|N'| = k$, for any agent $i \in N'$, there exists a $(\gamma,1,\dots,1)$-partition of size $k$ for agent $i$.

    According to Lemma \ref{lemma:whether-mms-or-reduce}, we can either get a $\gamma$-AMMS allocation on instance $\cI'$ or run Algorithm \ref{alg:reduced} to get a smaller reduced instance $\cI'$. If we get a $\gamma$-AMMS allocation, we are done. If we run Algorithm \ref{alg:reduced}, we can get a smaller instance $\cI'$ by removing at least one pair of agent and bundle. These two cases directly imply that the Algorithm~\ref{alg:n-agents} will terminate after $n$ rounds. According to the reduction rules in Algorithm \ref{alg:reduced}, every agent excluded from the reduced instance gets an MMS-feasible bundle. Thus we can conclude that the returned allocation of Algorithm \ref{alg:n-agents} is a $\gamma$-AMMS allocation for a general number of agents.
\end{proofof}

We next move to prove Lemma~\ref{lemma:gamma-partition}.
Depending on the size of the reduced instance, i.e., $|N'|$, we introduce different algorithms to compute a $(\gamma, 1, \ldots, 1)$-partition for some agent $i$.
The first algorithm is called Partition-Merging (see Algorithm~\ref{alg:MP}), in which we involve the MMS partition $P$ of agent $i$ for the original instance $\cI$.
Generally speaking, we consider the ``current state'' of the original MMS partition, i.e., $P'_j = P_j \cap M'$ for all $j\in N$.
We merge $n-k+1$ bundles with the smallest cost to be one, yielding a $k$-partition of items $M'$.
We show that when $|N'| > \lfloor \frac{n+1}{2} \rfloor$, the partition is a $(\gamma, 1, \ldots, 1)$-partition for agent $i$.

\begin{algorithm}[htbp]
        \caption{Partition-Merging($\cI, \cI', i$)} \label{alg:MP}
        \KwIn{An instance $\cI = (M, N, \bc)$, an instance $\cI' = (M', N', \bc')$ and an agent $i$}
        Let $k \gets |N'|$\;
        Let $P = \{P_1,\cdots,P_n\}$ be an MMS partition of  agent $i$ for instance $\cI$\; 
        Set $P'_j = P_j\cap M'$ for all $j\in N$\;
        Sort $\{P'_j\}_{j\in N}$ in the non-increasing order of $c_i(P'_j)$\;
        \KwOut{A partition $\{P'_1,\dots,P'_{k-1}, \bigcup_{j=k}^{n}P'_j\}$}
\end{algorithm}

\begin{lemma} \label{lemma:MP}
   Given a reduced instance $\cI' = (M',N',\bc')$ with $|N'| > \lfloor \frac{n+1}{2} \rfloor$, Algorithm~\ref{alg:MP} returns $P= \{P_1,\dots, P_k\}$ that is a $(\gamma,1,\dots,1)$-partition for agent $i$.
\end{lemma}

\begin{proof}
    Let $P = \{P_1,\cdots,P_n\}$ be an MMS partition of $M$ for agent $i$.
    Based on $P$, we construct a partition $P'$ of item $M'$ where $P'_j = P_j\cap M'$ for all $j\in N$. 
    We assume w.l.o.g. that $P'$ is sorted in non-increasing order of $c_i(P'_j)$ (see Figure~\ref{fig:MP} for an example).

    \begin{figure}[htbp]
    \centering
    \begin{tikzpicture}[scale = 0.75]					
        \begin{axis}[
        ybar stacked,
        enlargelimits=0.14,
        width = 0.8\textwidth,
	  height = 0.5\textwidth,
        bar width=0.4, 
        bar shift=0pt,
        xtick=data,
        ymin=0.15,
        legend style={draw=none},
        legend pos = north west,
	ylabel =  Bundle Cost,
        ] 
        
        	\addplot[draw=black, fill=gray!50] coordinates{
        		(1,0.70)
                    (2,0.65)
                    (3,0.60)
                    (4,0.55)
                    (5,0.50)
                    (6,0.45)
                    (7,0.40)
                    (8,0.4)
        	};
        	\addplot[draw=black,postaction={pattern=north east lines}, fill=gray!10] coordinates{
                    (1,0.28)
                    (2,0.33)
                    (3,0.38)
                    (4,0.43)
                    (5,0.48)
                    (6,0.53)
                    (7,0.58)
                    (8,0.6)
            
        	};
        
        \draw [densely dashed] (-1,1)--(10,1); 
        \end{axis}
        \draw[decorate,decoration={brace,raise=10pt,amplitude=0.15cm},black] (9.3,-0.3)--(5.8,-0.3);
        \node at (7.5,-1.2) {Combine the last $4$ bundles};
      \node at (5.1,1.7) {$P'_5$};
      \draw [<->] (5.4,2.5)--(5.4,0.0);
    \end{tikzpicture}
    \caption{An illustrating example of Partition-Merging with $n=8$ and $k=5$.}
    \label{fig:MP}
    \end{figure}
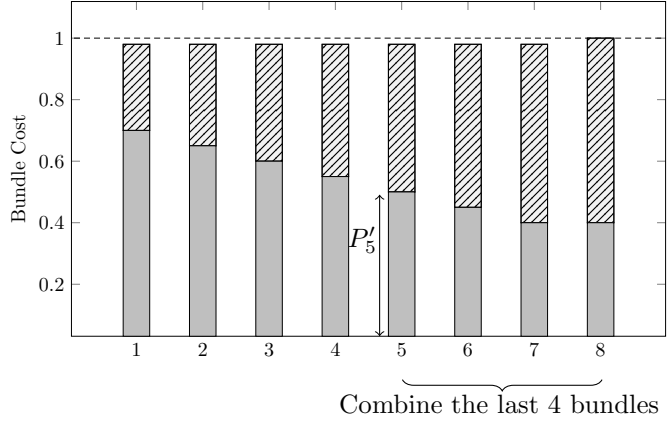

    By merging the last $|N|-|N'|+1$ bundles in $P'$, we have $\{P'_1,\dots,P'_{k-1}, \bigcup_{j=k}^{n}P'_j\}$ being a $k$-partition of $M'$.
    It is easy to verify that $c_i(P'_j)\le 1$ for every $j\in [k-1]$ since $c_i(P'_j)\le c_i(P_j)$ for all $j\in N'$. 
    For the last bundle $\bigcup_{j=k}^{n}P'_j$, we have
     \begin{equation*}
         c_i(\bigcup_{j=k}^{n}P'_j)\le  \frac{n-k+1}{n} \cdot \sum_{j=1}^n c_i(P'_j) <  \frac{n-k+1}{n} \cdot k= \frac{-k^2+(1+n)k}{n}, 
     \end{equation*}
     where the first inequality holds because $P'$ is sorted in non-increasing order of $c_i(P'_j)$ and the second inequality holds because $\sum_{j=1}^n c_i(P'_j) = c_i(M')<k$.

     Recall that $k > \lfloor \frac{n+1}{2} \rfloor$, we have
     \begin{equation*}
         c_i(\bigcup_{j=k}^{n}P'_j)  <\frac{-k^2+(1+n)k}{n} <\frac{(n+1)^2}{4n} = \gamma,
     \end{equation*}
     which completes our proof.
\end{proof}

In the following, we consider the case that $|N'| \leq \lfloor \frac{n+1}{2} \rfloor$, for which we propose the Capped-Bag-Filling Algorithm to compute a $(\gamma, 1, \ldots, 1)$-partition for some agent $i$.
At the beginning of the algorithm, we initialize $|N'|$ empty bundles.
We allocate items from the largest cost to the minimum, w.r.t. to $c_i$, while we only assign item $e$ to bundle $P$ if $c_i(P+e) \leq 1$. 
When it is not possible to allocate any item to any bundle, the remaining items will be assigned to the bundle with the lowest cost.

\begin{algorithm}[htb]
        \caption{Capped-Bag-Filling($\cI', i$)} \label{alg:CBF}
        \KwIn{An instance $\cI = (M', N', \bc')$ and an agent $i$}
        Let $k \gets |N'|$ \;
        Initialize $P = \{P_j\}_{j\in [k]}$ with $P_j \gets \emptyset$ \;
        Sort items with $c_i(e_1) \ge \dots \ge c_i(e_{|M'|})$\;
      
            \For{$j' = 1$ to $|M'|$}{
            \If{there exists a bundle $P_j$ such that $c_i(P_j+ e_{j'})\le 1$}{
                $P_j\gets P_j+ e_{j'}$; $M'\gets M'-e_{j'}$\;
            }
        }
        Sort $P$ with $c_i(P_1)\ge \dots\ge c_i(P_k) $\;
        $P_k\gets P_k\cup M'$ \;
        \KwOut{A partition $ \{P_1,\dots, P_k\}$}
\end{algorithm}

\begin{lemma} \label{lemma:CBF}
    Given a reduced instance $\cI' = (M',N',\bc')$ with $|N'| \le \lfloor \frac{n+1}{2} \rfloor$, Algorithm~\ref{alg:CBF} returns $P= \{P_1,\dots, P_k\}$ that is a $(\gamma,1,\dots,1)$-partition for agent $i$.
\end{lemma}

\begin{proof}
    According to Algorithm \ref{alg:CBF}, input an instance $\cI'=(M',N',\bc')$ and an arbitrary agent $i \in N'$, we get a $k$-partition of items $M'$, see Figure \ref{fig:CBF} as an example.
    Denote by $P'_k$ the items that are assigned to bundle $P_k$ in line 12 of Algorithm \ref{alg:CBF}. 
    For any item $e\in P'_k$, we have $c_i(P_j+e)>1$ for all $j\in \{1, \ldots, k-1\}$ following the algorithm.

    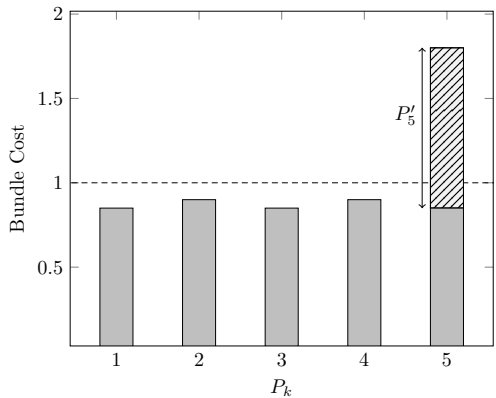
\begin{figure}[htbp]
    \centering
    \begin{tikzpicture}[scale = 0.75]						
        \begin{axis}[
        ybar stacked,
        enlargelimits=0.14,
        width = 0.6\textwidth,
	  height = 0.5\textwidth,
        bar width=0.4, 
        bar shift=0pt,
        xtick=data,
        ymin=0.25,
        legend style={draw=none},
        legend pos = north west,
        xlabel = $P_k$,
	ylabel =  Bundle Cost,
        ] 
        
        	\addplot[draw=black, fill=gray!50] coordinates{
        		(1,0.85)
                    (2,0.9)
                    (3,0.85)
                    (4,0.9)
                    (5,0.85)
        	};
        	\addplot[draw=black,postaction={pattern=north east lines}, fill=gray!10] coordinates{
                    (5,0.95)
            
        	};
        
        \draw [<->] (4.7,1.8)--(4.7,0.85);
        \node at (4.5,1.4) {$P'_5$};
        \draw [densely dashed] (-1,1)--(10,1); 
        \end{axis}
      
    \end{tikzpicture}
    \caption{An illustrating example of Capped-Bag-Filling for 5 agents.}
    \label{fig:CBF}
    \end{figure}
    
    We claim that the size of $P'_k$ can not exceed $k-1$, i.e., $|P'_k|\le k-1$. 
    Suppose by contradiction that $|P'_k|\ge k$, by reallocating the $|P'_k|$ items of bundle $P'_k$ to bundles $\{P_1,\cdots,P_k\}$ evenly, we have $\sum_{i\in [k]}c_i(P_i)>k$ which is a contradiction. 
    
    Denote by $e'_j$ the largest cost item in bundle $P_j$ for every $j\in [k]$. 
    Let $P = \bigcup_{j\in [k]}e'_j$ and $\hat{M} = P\cup P'_k$. 
    Since the items are allocated in descending order of their costs, we have $c_i(e'_1)\ge \cdots \ge c_i(e'_k)$ and $c_i(e'_k)\ge c_i(e) \;\forall e\in P'_k$. 
    Then we have 
    \begin{equation*}
        c_i(P'_k)\le \frac{|P'_k|}{|\hat{M}|}\cdot c_i(\hat{M})\le \frac{k-1}{2k-1}\cdot c_i(\hat{M}) \le  \frac{k-1}{2k-1}\cdot k.
    \end{equation*}
     where the last inequality holds due to $\hat{M}\subseteq M'$ and $c_i(M')\le k$. Due to line 11 of Algorithm \ref{alg:CBF}, we have $c_i(P_k\setminus P'_k)\le c_i(P_j)$ for every $j\in [k-1]$. Thus $c_i(P_k\setminus P'_k) \le \frac{k-c_i(P'_k)}{k}$. Then we have 
     \begin{equation*}
         c_i(P_k) = c_i(P'_k)+c_i(P_k\setminus P'_k) \le \frac{k+(k-1)\cdot c_i(P'_k)}{k} \le \frac{k+(k-1)\cdot \frac{k-1}{2k-1}\cdot k}{k} \le \frac{k^2}{2k-1}.
     \end{equation*}
    Since $k \le \lfloor \frac{n+1}{2} \rfloor $, we have 
    \begin{equation*}
         c_i(P_k) \le \frac{k^2}{2k-1} \le \frac{(n+1)^2}{4n} = \gamma,
     \end{equation*}
    which completes the proof.
\end{proof}

\section{Conclusion}
In this paper, we study the all-but-one MMS (AMMS) allocation where we guarantee all but one agent to achieve their maximin share. 
We show that the approximation ratio is $9/8$ for three agents and $4/3$ for four agents. 
For $n\ge 5$ agents, we provide an upper bound of $\frac{(n+1)^2}{4n}$ approximation.
%


The AMMS framework provides a novel approach to examining the approximation of MMS allocations from an optimization perspective. 
This alternative perspective may further deepen our understanding of the MMS concept. 
Moreover, AMMS naturally inherits all existing counterexamples of MMS, e.g., there is only one agent not assured of her MMS in the hard instance of~\cite{conf/wine/FeigeST21}. 
This indicates that AMMS is a natural and appropriate approximation method for MMS.
Given the minor gap between the hardness and MMS allocation, we propose the existence of $O(1)$-AMMS allocations for a general number of agents as an interesting open problem worth exploring.
Additionally, exploring whether our results for a small number of agents can be improved is also of interest. 
We also believe that the AMMS framework may inspire further research on other MMS relaxations, such as MMS allocations with item disposal or MMS allocations with subsidies.
It is noteworthy that any AMMS allocation can be converted into an MMS allocation with unallocated items (or with subsidies), only requiring to address the agent who is not guaranteed her MMS.



\bibliographystyle{abbrv}
\bibliography{arXiv}

\newpage
\appendix
\section{Complete Algorithms for Four Agents}\label{appendix:complete-algo}

\begin{algorithm}[H]
    \caption{$\frac{4}{3}$-AMMS-Four-Agents$(M, N, \bc)$}\label{alg:four-agents}
    \KwIn{An instance $\cI = (M, N, \bc)$ with $|N| = 4$}
    Initialize an MMS partition $P = \{ P_1, P_2, P_3, P_4 \}$ for agent $4$.\\
    Construct MMS-feasibility graph $G = (N \cup P, E)$.\\
    \eIf {there exists a perfect matching for  $G$} {
        Allocate each bundle to the matched agent based on perfect matching.\\
    }
    {
        Let $S^* \gets \argmax_{\{S \subseteq N: |S| > |L(S)| \}}|S|$.\\
        \uIf{$|S^*| = 2$}{
            Assume w.l.o.g. $S^* = \{ 1, 2 \}, L(S^*) = \{ P_1 \}$.\\
            Find the perfect matching for agents 3 and 4 and allocate the bundle to them respectively based on perfect matching.\\
            $M' \gets M \setminus (X_3 \cup X_4)$\;
            $\{B_1, B_2\}\gets$ Load-Balancing($M', \{1,2\}, c_1$)\;
            $X_2 \gets \argmin_{B\in \{B_1, B_2\} } c_2(B)$, $X_1\gets M'\setminus X_2$.\\
        }
        \ElseIf{$|S^*| = 3$} {
            \uIf{$|L(S^*)| = 2$}
            {
                Assume w.l.o.g. that $L(S^*) = \{ P_1, P_2 \}$, bundle $P_2$ is only liked by agent 3, $c_1(P_1) \leq c_1(P_2) \leq c_1(P_3) \leq c_1(P_4)$.\\
                \eIf{every agent likes only one bundle} {
                    $X_3 \gets P_2$, $X_4 \gets P_4$.\\
                    $\{B_1, B_2\}\gets$ Load-Balancing($P_1 \cup P_3, 2, c_1$).\\
                    $X_2 \gets \argmin_{B\in \{B_1, B_2\} } c_2(B)$, $X_1\gets P_1\cup P_3\setminus X_2$.\\
                }
                {
                    Assume agent 1 likes $P_1$ and $P_2$.\\
                    $X_4 \gets P_4$, $M' \gets M \setminus P_4$.\\
                    $B'_1 \gets P_2$, $\{B'_2, B'_3\}\gets$ Load-Balancing($P_1 \cup P_3, 2, c_1$).\\
                }
            }
            \ElseIf{$|L(S^*)| = 1$}{
                Let agent $1$ do an MMS partition $\mathbf{B} = \{ B_1, B_2, B_3, B_4 \}$.\\
                Let $b_{ij} = B_i \cap P_j, i, j \in \{ 1, 2, 3, 4 \}$.\\
                \eIf{there exists an atomic bundle $b_{ij}$ such that $c_1(P_j \setminus b_{ij}) > 1$} {
                    $X_4 \gets P_j$, $M' \gets M \setminus (P_j \cup B_i)$, $B'_1 \gets B_i \setminus b_{ij}$.\\
                    $\{B'_2, B'_3\}\gets$ Load-Balancing($M', 2, c_1$).\\
                }
                {
                    Assume w.l.o.g. that $c_1(P_1) \leq c_1(P_2) \leq c_1(P_3) \leq c_1(P_4)$.\\
                    $b_2 \gets \argmin_{i \in \{1, 2, 3, 4 \}}c_1(b_{i2})$, $b_3 \gets \argmin_{j \in \{1, 2, 3, 4 \}}c_1(b_{j3})$.\\
                    $X_4 \gets P_4$, $M' \gets M \setminus P_4$.\\
                    \eIf{$c_1(P_1) \leq \frac{4}{5}$} {
                        $B_1' \gets P_1 \cup b_2 \cup b_3, B_2' \gets P_2 \setminus b_2, B_3' \gets P_3 \setminus b_3$.\\
                    }
                    {
                        $B'_1 \gets P_1, B_2' \gets P_2 \cup b_3, B'_3 \gets P_3 \setminus b_3$.\\
                    }
                }
            }
            \tcp{Upon here we have $\{B'_1, B'_2, B'_3\}$ being a $(\frac{4}{3}, 1, 1)$-partition for agent $1$.}
            $(X_1, X_2, X_3) \gets \frac{4}{3}$-AMMS-Three-Agents$(\cI', B'_1, B'_2, B'_3)$\;
        }
    }
    \KwOut{An allocation $\bX = (X_1, X_2, X_3, X_4)$}
\end{algorithm}

\end{document}